\documentclass[11pt]{article}
\usepackage{graphics,epsfig}
\usepackage[latin1]{inputenc}
\usepackage[english,activeacute]{babel}
\usepackage[]{graphicx}
\usepackage{latexsym}
\usepackage{amsmath, amsthm, amsfonts, amssymb}
\usepackage{verbatim}

\newcommand{\nd}{\noindent}
\newtheorem{theorem}{Theorem}

\begin{document}

\newtheorem{theo}{Theorem}[section]
\newtheorem{definition}[theo]{Definition}
\newtheorem{lem}[theo]{Lemma}
\newtheorem{prop}[theo]{Proposition}
\newtheorem{coro}[theo]{Corollary}
\newtheorem{exam}[theo]{Example}
\newtheorem{rema}[theo]{Remark}
\newtheorem{example}[theo]{Example}
\newtheorem{principle}[theo]{Principle}
\newcommand{\ninv}{\mathord{\sim}}
\newtheorem{axiom}[theo]{Axiom}
\title{On the lattice structure of probability spaces in quantum mechanics}

\author{{\sc Federico Holik}$^{1}$ \ {\sc ,} \ {\sc César Massri}$^{2}$ \ {\sc
 ,} \ {\sc Leandro Zuberman}$^{1}$ \ {\sc
 ,} {\sc A. Plastino}$^{1,\,3,\,4}$}

\maketitle

\begin{center}

\begin{small}
1- Universidad Nacional de La Plata, Instituto
de F\'{\i}sica (IFLP-CCT-CONICET), C.C. 727, 1900 La Plata, Argentina \\
2- Departamento de Matem\'{a}tica - Facultad de Ciencias Exactas y
Naturales\\ Universidad de Buenos Aires - Pabell\'{o}n I, Ciudad
Universitaria \\ Buenos Aires, Argentina. Becario CONICET \\
3- Universitat de les Illes Balears and IFISC-CSIC, 07122 Palma de Mallorca, Spain \\
4- Instituto Carlos I de Fisica Teorica y Computacional and
Departamento de Fisica Atomica, Molecular y Nuclear, Universidad de
Granada, Granada, Spain
\end{small}
\end{center}

\vspace{1cm}

\begin{abstract}
\noindent Let $\mathcal{C}$ be the set of all possible quantum
states. We study the convex subsets of $\mathcal{C}$ with attention
focused on the lattice theoretical structure of these convex subsets
and, as a result,  find a framework capable of unifying several
aspects of quantum mechanics, including entanglement and Jaynes'
Max-Ent principle. We also encounter links with entanglement
witnesses, which leads to a new separability criteria expressed in
lattice language. We also provide an extension of a separability
criteria based on convex polytopes to the infinite dimensional case
and show that it reveals interesting facets concerning the
geometrical structure of the convex subsets. It is seen that the
above mentioned framework is also capable of generalization to any
statistical theory via the so-called convex operational models'
approach. In particular, we show how to extend the geometrical
structure underlying entanglement to any statistical model, an
extension which may be useful for studying correlations in different
generalizations of quantum mechanics.
\end{abstract}
\bigskip
\noindent

\begin{small}
\centerline{\em Key words: entanglement-quantum information-convex
sets -MaxEnt approach }
\end{small}

\bibliography{pom}
\section{Introduction}

\noindent In this work we will tackle the entanglement phenomenon
from a special viewpoint, that of regarding quantum states as
``probability measures" [see for example \cite{stulpe2001}], which
leads us to discuss convex sets of probability
measures. Quantum probabilities are of a very different nature
than that of classical ones. After setting preliminary mathematical
notions and notations in section \ref{s:preliminaries} (which may be
optionally complemented with appendix \ref{s:ApendixA}), we shortly
review the differences and definitions between classical and quantal
probabilities in section \ref{s:probabilities}. Preliminary matters
may be skipped by the reader familiarized with the quantum formalism in
infinite dimensions.

\noindent The existence of probability models of a very different
nature, of which the classical and the quantum instances are just
two particularly important examples of a wider family,  was one of
the motivations for the study of the so called operational or convex
approach (COM), one of the protagonists of our present discourse.
Refs.
\cite{Barnum-Wilce-2006,Barnum-Wilce-2009,Barnum-Wilce-2010,Beltrametti.Varadarajan-2000,Gudder-StatisticalMethods,Cattaneo-Gudder-1999}
deal with the subject of COMs, for which  states (understood as
probability measures) and their convex structure play a key role,
while other related quantities emerge in rather natural fashion.
Note that there exist generalizations of quantum mechanics,
including non-linear versions, that are axiomatized using the convex
structure of the set of states (see \cite{MielnikGQS},
\cite{MielnikTF}, and \cite{MielnikGQM}). The approach  treats in
geometrical fashion the statistical theory of systems, which
includes quantum and classical mechanics (and  several other
theories as well). In these generalized probabilistic models,
generalized observables are used. In the particular case of quantum
theory, one encounters the important notion of Positive Operator
Valued Measures (POVM's). We will review the COM approach as well as
POVM's in section \ref{s:COMapproach}. In section
\ref{s:introduction to QL} we will revisit the formal structure and
associated definitions of entanglement including the infinite
dimensional case.

\nd Lattices are other main character in our present discourse. They
have been studied in the context of quantum mechanics since the
seminal paper of von Newmann \cite{BvN} and characterize the
structure of the subspaces of the Hilbert space of a quantum system.
The paper of reference \cite{BvN} has motivated several
investigations in logics, philosophy \cite{Putnam}, foundations of
physics, and algebraic logic. In the particular case of the
foundations of quantum mechanics, several lines of investigation
have been developed. It is difficult to list all of them. We just
cite here \cite{mackey57,jauch,piron,kalm83,kalm86,vadar68,
vadar70,greechie81,gudderlibro78,giunt91,pp91,belcas81}. For a
complete bibliography see for example
\cite{dallachiaragiuntinilibro}, \cite{dvupulmlibro}, and
\cite{HandbookofQL}.

\noindent Part of the study of composite quantum systems was
developed in \cite{aertsdaub1, aertsdaub2,FR81}. As we will show in
sections \ref{s:New language} and \ref{s:Convex lattice} of this
work, the convex set of quantum states is endowed with canonical
lattice structures, which allow us to disclose a new structural
feature of quantum mechanics that, in particular, allows for an
extension to the infinite dimensional case of the standpoint
developed in \cite{extendedql,Holik-Massri-Ciancaglini-2010}. The
reader non familiarized with lattice theory may find it useful to
take glance at appendix \ref{s:ApendixB}.

\noindent Our new lattice structures are not only mathematical
curiosities. Instead, they are the key factor for achieving a
unifying viewpoint regarding several constructs linked to the
geometrical properties of the quantum set of states. As a first
example of their power, we will use this approach to provide a
generalization and reformulation of the celebrated Jaynes Max-Ent
principle \cite{Jaynes-1957a,Jaynes-1957b} to arbitrary COM's in
section \ref{s:Max-Ent} (see also \cite{Holik-Plastino-2011b} for
more developments), thus displaying an interesting convergence
between lattice theory and COM approaches.

\noindent In section \ref{s:EntanglementWittness} we restrict
ourselves to the finite dimensional case to study a particular
reformulation of entanglement witnesses in lattice theoretical
terms, which  provides new proofs of known results. Surprisingly
enough, these new proofs serve i) as the source of new abstract
entanglement criteria, which can be expressed in lattice ``format"
and ii) to study the volume of the space of separable states, a
theme to be tackled elsewhere. We added appendix \ref{s:ApendixC}
for refreshing mathematical notions indispensable  in this
respect.

\noindent We return in section \ref{s:The Relationship for convex}
to the infinite dimensional case to discuss the problem of
characterizing entanglement, both from the geometrical and
algebraic viewpoints. Emphasis will be put on maps that can be
defined between the lattice of the system and its subsystems. We
will  extend to infinite dimension a recently advanced, abstract
entanglement criterium \cite{Holik-Plastino-2011a} and study some
consequences thereof. In particular, we will underline an
interesting unifying characteristic of entanglement that  is known
to  hold for pure states and reads\\

\fbox{\parbox{6.0in}{\nd A \emph{pure} state is
separable$\Longleftrightarrow$ it is a product state
($\Longleftrightarrow$ the entropy of its reduced states is
minimal)}} \vskip 3mm

\noindent It is also well known that no such a simple statement is
valid for mixed states. Using both the lattice theoretical approach
and our criteria we will show that it is possible to suitably
generalize the above mentioned assertion from pure states to
arbitrary states, thus leaving the pure instance as a particular
case. In order to do so, we will introduce first the notion of
informational invariant (advanced in \cite{Holik-Plastino-2011a}).
This concept
\begin{itemize} \item is advantageously cast in purely geometrical
terms and holds for the infinite dimensional case as well,
uncovering non trivial geometric and algebraic properties, and
\item  also provides us with a simple, unifying abstract framework
to characterize separability properties of arbitrary states (not
only pure ones).
\end{itemize}

\noindent Furthermore, we will extend in section
\ref{s:PositiveMaps} some of our results to {\sf any} probabilistic
model via the COM approach. In particular, we discuss how the
geometrical structure found for the quantum case in section
\ref{s:The Relationship for convex} can be extended via the COM
approach to any probabilistic model. Such generalization is due to
the purely geometrical nature of our criteria, and may be useful to
define entanglement for theories more general than that of quantum
mechanics (for example, semiclassical models or non-linear versions
of quantum mechanics). Finally, in section \ref{s:Conclusions} some
conclusions are drawn.

\section{Preliminaries}\label{s:preliminaries}

\nd For a Hilbert space $\mathcal{H}$ of dimension $N > 2$ the set
of pure states forms a $(2N - 2)$-dimensional manifold, of measure
zero, in the $(N^2 - 2)-$dimensional boundary $\partial
\mathcal{C}_N$ of the set $\mathcal{C}_N$ of density matrices. The
set of mixed quantum states $\mathcal{C}_N$ consists of Hermitian,
positive matrices of size $N$, normalized by the trace condition,
that is

\begin{eqnarray}
\mathcal{C}_N=  \{\rho: \rho = \rho^{\dagger};\,\,\, \rho
\ge 0;\,\,\, tr(\rho) = 1;\,\,\, dim(\rho) = N\}.
\end{eqnarray}

\nd It  can be shown  for finite dimensional bipartite states that
there exist always a non-zero measure $\mu_s$ in the neighborhood
of separable states containing maximum uncertainty ones.  $\mu_s$
tends to zero as the dimension tends to infinity. Finally, for an
infinitely dimensional Hilbert space almost all states are
entangled (i. e., separable states are \emph{never dense})
\cite{Clifton-Halvorson-1999,Clifton-Halvorson-Kent-2000}.

\subsection{Notation}
 \nd Let us fix the notation to be employed here.  $\mathcal{P}(\mathcal{H})$ will denote the set of all
closed subspaces of a Hilbert space $\mathcal{H}$ (arbitrary
dimension), which are in a one to one correspondence with the
projection operators. Because of this one to one link, one usually
employs the notions of ``closed subspace'' and ``projector'' in
interchangeable fashion. An important construct is $\mathcal{A}$,
the set of bounded Hermitian operators on $\mathcal{H}$, while the
bounded operators on $\mathcal{H}$ will be denoted by
$\mathcal{B}(\mathcal{H})$. The projective Hilbert space
$\mathbf{C}\mathbf{P}(\mathcal{H})$ of a complex Hilbert space
$\mathcal{H}$ is the set of equivalence classes of vectors $v$ in
$\mathcal{H}$, with $v \ne 0$, given by $v \sim w$ when $v = \lambda
w$, with $\lambda$ a non-zero scalar. Here the equivalence classes
for $\sim$ are also called projective rays. A trace class operator
is a compact one for which a  finite trace may be defined
(independently of the choice of basis). We will appeal below to the
set $\mathcal{C}$ containing all positive, hermitian, and
trace-class (normalized to unity) operators in
$\mathcal{B}(\mathcal{H})$.

\noindent Let us remind the reader that a {\sf lattice } $\mathcal{L}$ is a partially ordered set (also
called a poset) in which any two elements $a$ and $b$ have a
unique supremum (the elements' least upper bound ``$a\vee b$";
called their join) and an infimum (greatest lower bound ``$a\wedge
b$"; called their meet). Lattices can also be characterized as
algebraic structures satisfying certain axiomatic identities.
Since the two definitions are equivalent, lattice theory draws on
both order theory and universal algebra. For additional
details, see Appendix \ref{s:ApendixB}. \vskip 3mm

\nd Let $\mathcal{H}$ be a separable Hilbert space of arbitrary
dimension representing a quantum system. As stated above, the
bounded operators on $\mathcal{H}$ will be denoted by
$\mathcal{B}(\mathcal{H})$. There are many topologies and relevant
subsets of $\mathcal{B}(\mathcal{H})$. In the literature,
$\mathcal A$ has denoted different subsets of
$\mathcal{B}(\mathcal{H})$. While mainly it denotes the Hermitian
operators, in some works it denotes the Hilbert-Schmidt operators
\cite{zyczkowski1998}. In this work we will use the following
notation

\begin{equation*}
\mathcal A =\{T\in \mathcal{B}(\mathcal{H}):T^\dag=T\}
\end{equation*}

\noindent Suppose that $T$ is a compact operator such that

\begin{equation}
\sum_{i\in I}\langle v_{i}|T v_{i}\rangle<\infty
\end{equation}

\noindent for all orthonormal basis $\{|v_{i}\rangle\}_{i\in I}$.
Then, the map $\mbox{tr}(\cdot)$ defined as

\begin{equation}
\mbox{tr}(T)=\sum_{i\in I}\langle v_{i}|T v_{i}\rangle
\end{equation}
\noindent is independent of the choice of basis. The set of
Hilbert Schmidt operators will be denoted by
$\mathcal{B}_2(\mathcal{H})$ and are defined by
\[\mathcal{B}_2=\{T\in \mathcal{B}(\mathcal{H}):\mbox{tr}(T^2)<\infty\}.\]
The space $\mathcal{B}_2$ endowed with the inner product $\langle
T_1,T_2\rangle=\mbox{tr}(T_2^\dag T_1)$ is a Hilbert space. For
$T\in\mathcal{B}(\mathcal{H})$ the absolute value of $T$ is
defined by $|T|=(T^\dag T)^{1/2}$. We can also consider the
subspace formed by the trace class operator, defined by \[\mathcal
B_1
=\{T\in\mathcal{B}(\mathcal{H}):|T|^{1/2}\in\mathcal{B}_2(\mathcal{H})\}.\]
It can be shown that the following statements are equivalent:

\begin{enumerate}
\item $T\in \mathcal{B}_1$.
\item $T=AB$ for $A,B\in \mathcal{B}_2(\mathcal{H})$.
\item $|T|\in\mathcal{B}_1(\mathcal{H})$.
\item $\mbox{tr}(|T|)<\infty$.
\end{enumerate}

\noindent The space of trace class operators is a Banach space
endowed with the norm $\|T\|=\mbox{tr}(|T|)$.

\noindent Notice that in $\mathcal{B}_2(\mathcal{H})\cap \mathcal
A$, the norm induced by the inner product is given by
$\|T\|=\mbox{tr}(T^\dag T)^{1/2}=\mbox{tr}(T^2)^{1/2}$ and it
coincides with the $\ell^2$ norm of the eigenvalues while in
$\mathcal{B}_1(\mathcal{H})\cap \mathcal A$ the norm coincides
with the $\ell ^1$ norm of the eigenvalues. Thus,
$\mathcal{B}_2\subset\mathcal{B}_1$ in $\mathcal A$.

\nd Coming back to the closed subspaces of $\mathcal{H}$ which are
in a one to one correspondence with the projection operators,  an
operator $P\in\mathcal{B}(\mathcal{H})$ is said to be a projector
if it satisfies

\begin{equation}
P^{2}=P
\end{equation}

\noindent and

\begin{equation}
P=P^{\dag}
\end{equation}

\subsection{Elementary measurements and projection
operators}\label{s:COMandEffects}

\nd A projection operator represents an elementary measurement
given by a yes-no experiment, i.e., a test in which we get the
answer ``yes" or the answer ``no". If $\mathcal{R}$ is the real
line, let $B(\mathcal{R})$ be the family of subsets of
$\mathcal{R}$ such that

\begin{itemize}
\item 1 - The family is closed under set theoretical complements.

\item 2 - The family is closed under denumerable unions.

\item 3 - The family includes all open intervals.
\end{itemize}

\noindent The elements of $B(\mathcal{R})$ will be called the
\emph{Borel subsets} of $\mathcal{R}$ \cite{ReedSimon}. A
projection valued measure (PVM) $M$, is a mapping

\begin{subequations}

\begin{equation}
M: B(\mathcal{R})\rightarrow \mathcal{P}(\mathcal{H})
\end{equation}

\noindent such that

\begin{equation}
M(0)=0
\end{equation}
\begin{equation}
M(\mathcal{R})=\mathbf{1}
\end{equation}
\begin{equation}
M(\cup_{j}(B_{j}))=\sum_{j}M(B_{j}),\,\,
\end{equation}

\noindent for any disjoint family ${B_{j}}.$ Also,

\begin{equation}
M(B^{c})=\mathbf{1}-M(B)=(M(B))^{\bot}
\end{equation}

\end{subequations}

\noindent  All operators representing observables may be expressed
in terms of projection operators (and so by elementary measurements)
via the spectral decomposition theorem, which asserts that the set
of spectral measurements may be put in a bijective correspondence
with the set $\mathcal{A}$ of Hermitian operators of $\mathcal{H}$.
A list of set-theory concepts used in this work can be found in
Appendix \ref{s:ApendixA}

\nd Elementary (sharp) tests in quantum mechanics are represented by
projection operators that form the well known von Newmann's lattice
${\mathcal{L}}_{vN}$, an orthomodular one (see Appendix
\ref{s:ApendixB}). The Born-rule implies that probabilities in
quantum mechanics are linked to measures over the von Newmann's
lattice. Using Gleason's theorem, it is possible to link in a
bijective way density matrixes and non-kolmogorovian probability
measures (more on this below).

\section{Quantum vs. classical probabilities}\label{s:probabilities}

\noindent  The reader is advised to consult the Appendixes regarding
some mathematics concepts appealed to below. It is a well known
fact, since the 30's, that a quantum system represented by a Hilbert
space $\mathcal{H}$ is  associated to a lattice formed by all its
closed subspaces \newline
${\mathcal{L}}_{v\mathcal{N}}({\mathcal{H}})=
<{\mathcal{P}}({\mathcal{H}}),\ \cap,\ \oplus,\ \neg,\ 0,\ 1>$,
where $0$ is the empty set $\emptyset,$  $1$  the total space
$\mathcal{H}$, $\cap$  the intersection, $\oplus$ the closure of the
sum, and $\neg(\mathcal{S})$  the orthogonal complement of a
subspace $\mathcal{S}$ \cite{mikloredeilibro}. This is the Hilbert
lattice, named ``Quantum Logic" by Birkhoff and von Neumann
\cite{BvN}. We will refer to this lattice as
${\mathcal{L}}_{v\mathcal{N}}$, the `von Neumann lattice'. Thus, the
set of elementary yes-no tests has an orthomodular lattice
structure, which is itself non-boolean, modular in the finite
dimensional case, and never modular in the infinite one. We will
relate below elementary tests to quantum probability spaces and
study their lattice structure.

\noindent Given a set $\Omega$, let us consider a $\sigma$-algebra
(see  Appendixes) $\Sigma$ of $\Omega$. Then, a probability measure
will be given by a function $\mu$ such that

\begin{subequations}\label{e:kolmogorovian}
\begin{equation}
\mu:\Sigma\rightarrow[0,1]
\end{equation}
\noindent which satisfies
\begin{equation}
\mu(\emptyset)=0
\end{equation}
\begin{equation}
\mu(A^{c})=1-\mu(A),
\end{equation}

\noindent where $(\ldots)^{c}$ means set-theoretical-complement
and for any pairwise disjoint denumerable family $\{A_{i}\}_{i\in
I}$

\begin{equation}
\mu(\bigcup_{i\in I}A_{i})=\sum_{i}\mu(A_{i})
\end{equation}

\end{subequations}

\noindent where conditions (\ref{e:kolmogorovian}) are the well
known axioms of Kolmogorov.

\nd  In the formulation of both classical and quantum probabilities,
states can be regarded as representing consistent probability
assignments \cite{wilce}. In the quantum mechanics instance {\it
this ``states as mappings" visualization}  is achieved via a
function \cite{mikloredeilibro}

\begin{subequations}
\begin{equation}\label{e:nonkolmogorov}
s:\mathcal{P}(\mathcal{H})\rightarrow [0;1]
\end{equation}

\noindent such that:

\begin{equation}
s(\textbf{0})=0 \,\, (\textbf{0}\,\, \mbox{is the null subspace}).
\end{equation}

\begin{equation}
s(P^{\bot})=1-s(P),\end{equation} \noindent and, for a denumerable and
orthogonal family of projections

\begin{equation}
\,\, {P_{j}}, \,\,s(\sum_{j}P_{j})=\sum_{j}s(P_{j}).
\end{equation}

\end{subequations}

\noindent  The above equation defines a probability, but in fact,
not a classical one, because classical probability axioms obey the
Kolmogorov's axioms of equation (\ref{e:kolmogorovian}). The main
difference comes from the fact that the $\sigma$-algebra in
(\ref{e:kolmogorovian}) is boolean, while
$\mathcal{P}(\mathcal{H})$ is not. Thus, quantum probabilities are
also called non-kolmogorovian (or non-boolean) probability
measures. The crucial fact is that, in the quantum case, we do not
have a $\sigma$-algebra, but an orthomodular lattice of
projections. \noindent  Most importantly, Gleason's theorem
\cite{Gleason,Gleason-Dvurechenski-2009} asserts that if
$dim(\mathcal{H})\geq 3$, then: \vskip 3mm

\fbox{\parbox{6.0in}{\nd The set of all measures of the form
(\ref{e:nonkolmogorov}) can be put into one to one correspondence
with the set $\mathcal{C}$ formed by all positive, hermitian and
trace-class (normalized to unity) operators in
$\mathcal{B}(\mathcal{H})$}} \vskip 3mm

\nd More generally, consider a $C^{\ast}$-algebra $\mathbf{A}$.  The
prototypical example of a such an algebra is the algebra
$\mathcal{B}(\mathcal{H})$ of bounded (equivalently continuous)
linear operators defined on a complex Hilbert space. In general, a
state $\varphi$ will be a positive linear functional of norm equal
to unity. If the algebra has a unit (as is the case for
$\mathcal{B}(\mathcal{H})$), states will be given by the
intersection of the closed affine hyperplane $\varphi(\mathbf{1})=1$
and the set of positive linear forms on $\mathbf{A}$ of norm $\leq
1$ (which is compact in the topology of pointwise convergence) and
then the concomitant  extension of the set $\mathcal{C}$ will be a
\underline{convex and compact space}. If $P\in\mathcal{P}(\mathcal{H})$ the
correspondence between $\rho\in\mathcal{C}$ and its induced
probability measure is given by

\begin{equation}\label{e:bornrule}
s_{\rho}(P)=\mbox{tr}(\rho P)
\end{equation}

\noindent Equation (\ref{e:bornrule}) is essentially Born's rule.
Any $\rho\in\mathcal{C}$ may be written as

\begin{equation}\label{e:convexity}
\rho=\sum_{i}p_{i}P_{\psi_{i}}
\end{equation}

\noindent where the $P_{\psi_{i}}$ are one dimensional projection
operators on the rays (subspaces of dimension one) generated by
the vectors $\psi_{i}$ and $\sum_{i}p_{i}=1$ ($p_{i}\geq 0$).
Thus, it is clear that $\mathcal{C}$ is a convex set. If the sum
in (\ref{e:convexity}) is finite, then $\rho$ is said to be of
finite range. It is important to remark that in the infinite
dimensional case, the sum in (\ref{e:convexity}) may be infinite
in a non-trivial sense. $\mathcal{C}$ is then a set formed of
non-boolean probability measures. That $\mathcal{C}$ is a closed
convex set can also be seen by the fact that if we define the
half-planes

\begin{subequations}
\begin{equation}
H_x=\{\rho\in\mathcal{A}\,|\,x^\dag\rho x < 0\}
\end{equation}
\begin{equation}
H_x^+=\{\rho\in\mathcal{A}\,|\,x^\dag\rho x\geq0\}
\end{equation}
\end{subequations}

\noindent then

\begin{equation}
\mathcal{C}=\bigcap_{x\in\mathcal{H}}H_x^+\cap\{\rho\,|\,\text{tr}(\rho)=1\}
\end{equation}

\noindent If we consider now $\mathcal{L}_{\mathcal{C}}$ as the
set of all convex subsets of $\mathcal{C}$, that is

\begin{definition}
$\mathcal{L}_{\mathcal{C}}=\{C\subseteq\mathcal{C}\,\,|\,\,C\,\,\mbox{is
convex}\}$
\end{definition}

\noindent then any element of $\mathcal{L}_{\mathcal{C}}$ will be
itself a ``probability space", in the sense that it is a set of
non boolean probability measures closed under convex combinations
(not to be confused with the usual mathematical notion of sample
space). We will show that $\mathcal{L}_{\mathcal{C}}$ is endowed
with a canonical lattice structure in Section \ref{s:Convex
lattice}.

\noindent A general (pure) state can be written as (using Dirac's
notation):

\begin{equation}
\rho=|\psi\rangle\langle\psi|
\end{equation}

\noindent We will denote the set of all pure states by

\begin{equation}
P(\mathcal{C}):=\{\rho\in\mathcal{C}\,|\, \rho^{2}=\rho\}
\end{equation}
This set is in correspondence with the rays of $\mathcal{H}$ via
the association:

\begin{equation}
\mathcal{F}:\mathcal{C}\mathcal{P}(\mathcal{H})\rightarrow
\mathcal{C}, \quad [|\psi\rangle]\mapsto|\psi\rangle\langle\psi|
\end{equation}

\noindent where $\mathcal{C}\mathcal{P}(\mathcal{H})$ is the
projective space of $\mathcal{H}$, and $[|\psi\rangle]$ is the
class defined by the vector $|\psi\rangle$
($|\varphi\rangle\sim|\psi\rangle\longleftrightarrow|\varphi\rangle=\lambda|\psi\rangle$,
$\lambda\neq 0$). If $M$ represents an observable, its mean value
$\langle M\rangle$ is given by

\begin{equation}\label{e:meanvalueoperator}
\mbox{tr}(\rho M)=\langle M\rangle
\end{equation}

\noindent Notice that the set of positive operators forms a cone,
and that the set of trace class operators (of trace one) forms an
hyperplane. Thus, $\mathcal{C}$ is the intersection of a cone and
an hyperplane embedded in $\mathcal{A}$. This structure (or
geometrical convex setting) is susceptible of considerable
generalization (see
\cite{Barnum-Wilce-2006,Barnum-Wilce-2009,Barnum-Wilce-2010} for
an excellent overview, partly reproduced in this work for the sake
of completeness).

\section{Quantal effects and convex operational approach
(COM)}\label{s:COMapproach}

\nd In modelling probabilistic operational theories one associates
to any probabilistic system a triplet $(X,\Sigma,p)$, where
\begin{enumerate}
\item $\Sigma$ represents the set of states of the system,
\item $X$ is the set of possible measurement outcomes, and
\item $p:X\times \Sigma\mapsto [0,1]$ assigns to each outcome $x\in
X$ and state $s\in\Sigma$ a probability $p(x,s)$ of $x$ to occur
if the system is in the state $s$.  \item   If we fix $s$ we
obtain the mapping $s\mapsto p(\cdot,s)$ from $\Sigma\rightarrow
[0,1]^{X}$.
\end{enumerate}

\nd Note that
\begin{itemize}
\item This again identifies all the states of $\Sigma$ with maps.
\item Considering their closed convex hull, we obtain the set
$\Omega$ of possible probabilistic mixtures (represented
mathematically by convex combinations) of states in $\Sigma$.
\item In this way one also obtains, for any outcome $x\in X$, an
affine evaluation-functional $f_{x}:\Omega\rightarrow [0,1]$,
given by $f_{x}(\alpha)=\alpha(x)$ for all $\alpha\in \Omega$.
\item  More generally, any affine functional $f:\Omega\rightarrow
[0,1]$ may be regarded as representing a measurement outcome and
thus use $f(\alpha)$ to represent the probability for that outcome
in state $\alpha$.
\end{itemize}

\nd  For the special case of quantum mechanics, the set of all
affine functionals so-defined are called effects. They form an
algebra (known as the \emph{effect algebra}) and represent
generalized measurements (unsharp, as opposed to sharp measures
defined by projection valued measures). The specifical form of an
effect in quantum mechanics is as follows. A generalized
observable or \emph{positive operator valued measure} (POVM)
\cite{Busch-Lahti-2009,Thesis-Heinonen-2005,Ma-Effects} will be
represented by a mapping

\begin{subequations}

\begin{equation}
E:B(\mathcal{R})\rightarrow\mathcal{B}(\mathcal{H})
\end{equation}

\noindent such that

\begin{equation}
E(\mathcal{R})=\mathbf{1}
\end{equation}

\begin{equation}
E(B)\geq 0, \,\,\mbox{for any}\,\, B\in B(\mathcal{R})
\end{equation}

\begin{equation}
E(\cup_{j}(B_{j}))=\sum_{j}E(B_{j}),\,\, \mbox{for any disjoint
familly}\,\, {B_{j}}
\end{equation}

\end{subequations}

\noindent The first condition means that $E$ is normalized to
unity, the second one that $E$ maps any Borel set B to a positive
operator, and the third one that $E$ is $\sigma$-additive with
respect to the weak operator topology. In this way, a generalized
POVM can be used to define a family of affine functionals on the
state space $\mathcal{C}$ (which corresponds to $\Omega$ in the
general probabilistic setting) of quantum mechanics as follows

\begin{subequations}

\begin{equation}
E(B):\mathcal{C}\rightarrow [0,1]
\end{equation}

\begin{equation}
\rho\mapsto \mbox{tr}(E\rho)
\end{equation}

\end{subequations}

\noindent Positive operators $E(B)$ which satisfy $0\leq
E\leq\mathbf{1}$ are called effects (which form an \emph{effect
algebra} \cite{Cattaneo-Gudder-1999,EffectAlgebras-Foulis-2001}).
Let us denote by $\mathrm{E}(\mathcal{H})$ the set of all effects.\vskip 3mm \nd
 {\sf Indeed, a  POVM
 is a measure whose values are non-negative self-adjoint
operators on a Hilbert space. It is the most general formulation
of a measurement in the theory of quantum physics}.  \vskip 3mm \nd A rough
analogy would consider that a POVM is to a projective measurement
what a density matrix is to a pure state. Density matrices can
describe part of a larger system that is in a pure state
(purification of quantum state); analogously, POVMs on a physical
system can describe the effect of a projective measurement
performed on a larger system. Another, slightly different way to
define them is as follows:

\nd Let $(X, M)$ be measurable space; i.e., $M$ is a
$\sigma-$algebra of subsets of $X$. A POVM is a function $F$ defined
on $M$ whose values are bounded non-negative self-adjoint operators
on a Hilbert space $\mathcal{H}$ such that $F(X) = I_H$ (identity)
and for every i) $\xi \in  \mathcal{H}$ and ii) projector $P=
|\psi\rangle\langle  \psi|;\,\, |\psi\rangle \in \mathcal{H}$,  $P
\rightarrow\, \langle F(P)\xi \vert  \xi \rangle$ is a  non-negative
countably additive measure on $M$. This definition should be
contrasted with that for the projection-valued measure, which is
very similar, except that, in the projection-valued measure, the
$F$s are required to be projection operators.

\subsection{Convex operational approach}

\nd Returning now to the general model of probability states we
may, as explained in the Appendix, consider the convex set
$\Omega$ as the basis of a positive cone $V_{+}(\Omega)$ of the
linear space $V(\Omega)$. Thus, every affine linear functional can
be extended to a linear functional in $V(\Omega)^{\ast}$ (the dual
linear space). It can be shown that there is a unique unity
functional such that $u_{\Omega}(\alpha)=1$ for all
$\alpha\in\Omega$ (in quantum mechanics, this unit functional is
the trace function). The general operational or convex approach
may be viewed as the triplet $(A,A^{\sharp},u_{A})$, where
\begin{enumerate} \item $A$ is a normed space endowed with \item a strictly positive
linear functional $u_{A}$, and \item $A^{\sharp}$ is a weak-$\ast$
dense subspace of $A^{\ast}$ ordered by a chosen regular cone
$A_{+}^{\sharp}\subseteq A_{+}^{\ast}$ containing $u_{A}$.
\item Effects will be given by functionals $f$ in $A_{+}^{\sharp}$
such that $f\leq u_{A}$. \end{enumerate} \nd As viewed  from the
COM's standpoint, one of the characteristic features which
distinguish classical from quantum mechanics is the fact that in
the classical instance any non-pure state has a unique convex
decomposition in pure states, while this is no longer true in the
quantum case.   The forthcoming section reviews the principal
features of the convex set of states and relates them to
entanglement. This will be useful for studying  canonical maps
 defined between the probability spaces of a composite
system and those of its subsystems. Afterwards, in subsequent
 sections,  {\it we will i) show that $\mathcal{L}_{\mathcal{C}}$ is
endowed with a canonical lattice theoretical structure, ii) find
its main features,  and iii) relate them to quantum entanglement
and positive maps}. Of course, some
 of these results can easily be extended to the general setting of
convex operational models.

\section{Entanglement and the convex set of states: an overview}\label{s:introduction to QL}

\nd Consider composite quantal systems $S$ of subsystems $S_{1}$ and
$S_{2}$, with associated  separable Hilbert spaces $\mathcal{H}_{1}$
and $\mathcal{H}_{2}$. The pure states  are given by rays in the
tensor product space
$\mathcal{H}=\mathcal{H}_{1}\otimes\mathcal{H}_{2}$. It is not true
that any pure state of $S$ factorizes after the interaction into
pure states of the subsystems, a  situation  very different to that
of classical mechanics, where for state spaces $\Gamma_{1}$ and
$\Gamma_{2}$ we assign $\Gamma=\Gamma_{1}\times\Gamma_{2}$ for their
composition. This is an absolutely central issue.

\vskip 3mm \nd Let us now briefly review the quantal relationship
between the $S-$states and the states of the subsystems. Consider
for simplicity  the bipartite case with systems $S_{1}$ and $S_{2}$.
If $\{|x_{i}^{(1)}\rangle\}$ and $\{|x_{i}^{(2)}\rangle\}$ are the
corresponding orthonormal basis of $\mathcal{H}_{1}$ and
$\mathcal{H}_{2}$, respectively, then the set
$\{|x_{i}^{(1)}\rangle\otimes|x_{j}^{(2)}\rangle\}$ constitutes an
orthonormal basis for $\mathcal{H}_{1}\otimes\mathcal{H}_{2}$.

\noindent For observables of the form $A_{1}\otimes\mathbf{1}_{2}$
and $\mathbf{1}_{1}\otimes A_{2}$ (with $\mathbf{1}_{1}$ and
$\mathbf{1}_{2}$ the identity operators in $\mathcal{H}_{1}$ and
$\mathcal{H}_{2}$ respectively), then reduced state operators
$\rho_{1}$ and $\rho_{2}$ can be defined for systems $S_{1}$ and
$S_{2}$ such that

\begin{equation}\label{e:conditiononreducedstate}
\mbox{tr}(\rho A_{1}\otimes\mathbf{1}_{2})=\mbox{tr}(\rho_{1}A_{1})
\end{equation}

\noindent with an analogous equation for $\rho_{2}$. This
assignation of reduced states is done via \emph{partial traces
maps}. It is of interest for us to study how to define these maps in
detail. In order to do so and given a state $\rho\in\mathcal{C}$,
consider the functional

\begin{eqnarray}
F_{\rho}:\mathcal{A}(\mathcal{H}_{1})\otimes\mathbf{1}_{2}\longrightarrow \mathcal{R}\nonumber\\
A\otimes\mathbf{1}_{2}\mapsto\mbox{tr}(\rho A\otimes\mathbf{1})
\end{eqnarray}

\noindent Clearly $F_{\rho}$ induces a map

\begin{eqnarray}
f_{\rho}:\mathcal{A}(\mathcal{H}_{1})\longrightarrow \mathcal{R}\nonumber\\
A\mapsto\mbox{tr}(\rho A\otimes\mathbf{1})
\end{eqnarray}

\noindent which specifies a state as defined in
(\ref{e:nonkolmogorov}) when restricted to projections in
$\mathcal{A}(\mathcal{H}_{1})$. Thus, using Gleason's theorem, there
exists $\rho_{1}$ such that $f_{\rho}(A)=\mbox{tr}(\rho_{1}A)$
(where now the trace is taken on Hilbert space $\mathcal{H}_{1}$).
We have thus shown that for any $\rho\in\mathcal{C}$ there exists
$\rho_{1}\in\mathcal{C}_{1}$ such that for any
$A_{1}\in\mathcal{A}(\mathcal{H}_{1})$ Equation
(\ref{e:conditiononreducedstate}) is satisfied.

\noindent This assignment is the one which allows to define the
partial trace $\mbox{tr}_{2}(\cdot)$ given by equation
(\ref{e:partialtrace}) below. An analogous reasoning leads to
$\mbox{tr}_{1}(\cdot)$. Accordingly, we can consider the maps:

\begin{subequations}\label{e:partialtrace}

\begin{equation}
\mbox{tr}_{i}:\mathcal{C}\longrightarrow \mathcal{C}_{j}
\end{equation}

\begin{equation}
\rho\mapsto\mbox{tr}_{i}(\rho)=\rho_{j}
\end{equation}

\end{subequations}

\noindent In
\cite{Clifton-Halvorson-1999,Clifton-Halvorson-Kent-2000} it is
shown that these maps are continuous and onto. \emph{Thus, partial
traces defined in equation (\ref{e:partialtrace}) are continuous and
onto, a fact that will be used in section \ref{s:The Relationship
for convex}}.

\noindent Operators of the form $A_{1}\otimes\mathbf{1}_{2}$ and
$\mathbf{1}_{1}\otimes A_{2}$ represent magnitudes related to
$S_{1}$ and $S_{2}$, respectively. When $S$ is in a product state
$|\varphi_{1}\rangle\otimes|\varphi_{2}\rangle$, the mean value of
the product operator $A_{1}\otimes A_{2}$ will be

\begin{equation}
\mbox{tr}(|\varphi_{1}\rangle\otimes|\varphi_{2}\rangle\langle\varphi_{1}
|\otimes\langle\varphi_{2}|A_{1}\otimes A_{2})=\langle
A_{1}\rangle\langle A_{2}\rangle
\end{equation}

\noindent reproducing statistical independence. Separable states are
defined (see \cite{Werner}) as those states of $\mathcal{C}$ which can be
approximated by a succession of states written as a convex
combination of product states of the form

\begin{equation}
\rho_{Sep}=\sum_{i,j}\lambda_{ij}\rho_{i}^{(1)}\otimes\rho_{j}^{(2)}
\end{equation}

\noindent where $\rho_{i}^{(1)}\in\mathcal{C}_{1}$ and
$\rho_{j}^{(2)}\in\mathcal{C}_{2}$, $\sum_{i,j}\lambda_{ij}=1$ and
$\lambda_{ij}\geq 0$. Thus, the set $\mathcal{S}(\mathcal{H})$ of
separable states is defined as the convex hull of all product
states closed in the trace norm topology (i.e., the trace norm
topology given by
$\|\rho\|_{\mbox{tr}}=\mbox{tr}((\rho^{\dag}\rho)^{\frac{1}{2}}$).
Accordingly,

\begin{definition}
$\mathcal{S}(\mathcal{H}):=\{\rho\in\mathcal{C}\,|\,\rho\,\,
\mbox{is separable}\}$
\end{definition}

\noindent There exist very many non-separable states in
$\mathcal{C}$, called  entangled ones
\cite{bengtssonyczkowski2006}. They form a set we will call
$\mathcal{E}(\mathcal{H})$, i.e.,

\begin{definition}
$\mathcal{E}(\mathcal{H})=\mathcal{C}\backslash\mathcal{S}(\mathcal{H})$
\end{definition}

\noindent Alike entangled states, separable states are reproducible
by local classical devices (but this by no means implies that they
\emph{are} classical). It is a fact that, equipped with the
trace-norm distance, $\mathcal{S}(\mathcal{H})$ is a complete
separable metric space. Also that a sequence of quantum states
converging to a state in the weak operator topology converges to it
as well in the trace norm. Stated in a more technical form, a state
in $\mathcal{B}(\mathcal{H}_{1}\otimes\mathcal{H}_{2})$ is called
separable if it is in the convex closure of the set of all product
states in $\mathcal{S}(\mathcal{H}_{1}\otimes\mathcal{H}_{2})$.

\noindent Estimating the volume of $\mathcal{S}(\mathcal{H})$ is
of great interest (see --among others--\cite{zyczkowski1998},
\cite{horodecki2001} and \cite{aubrum2006}). \noindent  For pure
states we have at hand the superposition principle

\begin{principle}\label{e:superposition principle}
\noindent Superposition Principle. If $|\psi_{1}\rangle$ and
$|\psi_{1}\rangle$ are physical states, then
$\alpha|\psi_{1}\rangle+\beta|\psi_{1}\rangle$
($|\alpha|^{2}+|\beta|^{2}=1$) will be a physical state too.
\end{principle}

\noindent Additionally, when we include mixtures we have

\begin{principle}\label{e:mixing principle}
\noindent  Mixing Principle. If $\rho$ and $\rho'$ are physical
states, then $\alpha\rho+\beta\rho'$ ($\alpha+\beta=1$,
$\alpha,\beta\geq 0$) will be a physical state too.
\end{principle}

\noindent There exist many  studies which concentrate on mixtures.
For example, this is the case for works on quantum decoherence
\cite{Schlosshauer,Castagnino-2011,Castagnino-2008}, quantum
information processing, or generalizations of quantum mechanics
which emphasize its convex nature (not necessarily equivalent to
``Hilbertian" $QM$). The set of interest in such studies is
$\mathcal{C}$ and not the lattice of projections. Consequently, it
seems adequate to consider in some detail the notion of structures,
including improper mixtures
\cite{extendedql,d'esp,Kirkpatrik,ReplytoKirkpatrik} and pure
states, acknowledging their important place in the physical
``discourse", as  anticipated in the Introduction, in the hope that
such structures will provide a natural framework for studying
foundational issues.

\noindent A good question to ask is under which conditions a given
state may be decomposed in terms of other states, separable states
being a particular case of decomposing a given state in terms of
product states. This is the subject of decomposition theory (see
\cite{OlaBratteli}, chapter four). Given a $C^{\ast}$-algebra
$\mathbf{A}$ with identity $\mathbf{1}$, the set of states forms a
convex compact space of the dual space $\mathbf{A}^{\ast}$
(compact in the weak$^{\ast}$ topology). Given a state $\rho$ such
that $\rho\in K$, with $K$ a closed convex set, one attempts to
write an expression of the form

\begin{equation}\label{e:baricentric}
\rho=\int_{K} d\mu(\rho')\rho',
\end{equation}

\noindent where $\mu$ is a measure supported by the set of extremal
points of $K$ (see Appendix). Equation (\ref{e:baricentric}) is
referred to as the \emph{barycentric decomposition of $\rho$}. The
example of interest refers to the $C^{\ast}$-algebra
$\mathcal{B}(\mathcal{H})$, the relevant set of states being
$\mathcal{C}$, while $K=\mathcal{S}(\mathcal{H})$. In this specific
instance, and according to the above definition, it is possible to
show that \cite{Kholevo-Schirikov-Werner}

\begin{equation}
\rho=\int_{\mathcal{C}_{1}}\int_{\mathcal{C}_{2}}\rho_{\mathcal{H}_{1}}\otimes\rho_{\mathcal{H}_{2}}\nu(d\rho_{\mathcal{H}_{1}}
d\rho_{\mathcal{H}_{1}})
\end{equation}

\noindent and also that

\begin{equation}\label{e:integralpureproduct}
\rho=\int_{P(\mathcal{C}_{1})}\int_{P(\mathcal{C}_{2})}|\psi\rangle\langle\psi|\otimes|\varphi\rangle\langle\varphi|\nu(d\psi
d\varphi)
\end{equation}

\noindent Before concentrating attention on more general convex
subsets of $\mathcal{C}$, we will restrict ourselves in next
section to  a special mathematical construct, closely linked to $\mathcal{P}(\mathcal{H})$.

\section{The Lattice $\mathcal{L}$}\label{s:New language}

\nd  On  $\mathcal{H}$ we find, among many, two particularly
important constructs, namely, i) $\mathcal{A}$, the set of bounded
{\it and} Hermitian operators  together with ii)  the set of
(``merely") bounded operators (they map bounded sets to bounded
sets),  denoted by $\mathcal{B}(\mathcal{H})$. This set of bounded
linear operators, together with the addition and composition
operations, the norm and the adjoin operation, is a C*-algebra
(see Appendix \ref{s:ApendixA}). We {\it begin now to present our
new results}.

\vskip 3mm

\nd  Let us start by  extending some developments and definitions
of \cite{extendedql} to Hilbert spaces of arbitrary (possibly
infinite) dimension. Let us define $G(\mathcal{A})$ as the lattice
associated to the pair $(\mathcal{A},\mbox{tr})$
\begin{equation}
G(\mathcal{A}):=\{S\subset \mathcal{A}\,|\, S \text{ is a
\underline{closed} }\mathcal{R}\text{-subspace}\}.
\end{equation}
\noindent It is well-known that $G(\mathcal{A})$ is a modular, orthocomplemented, atomic
and complete lattice (see Appendix \ref{s:ApendixB}). It is not distributive and hence not a Boolean
algebra. Let $\mathcal{L}$ be the associated, induced lattice in
$\mathcal{C}$:

\begin{equation}
\mathcal{L}:=\{S\cap\mathcal{C}\,|\, S\in G(\mathcal{A})\}
\end{equation}

\noindent for which we will define now canonical lattice
operations.

\subsection{Characterizing $\mathcal{L}$}
\noindent Note that there are many subspaces $S^{\,'}\in
G(\mathcal{A})$ such that $S\cap \mathcal{C}=S^{\,'}\cap
\mathcal{C}$. In order to deal with them and appropriately define
the lattice induced in $\mathcal{C}$ by $G(\mathcal{A})$, we need a
subspace that will represent the set $S\cap \mathcal{C}$. For each
$L\in\mathcal{L}$ we will choose as this representative that
subspace possessing the minimum dimension

\begin{equation}
S_{L}=\bigcap\{S\in G(\mathcal{A})\,|\,S\cap\mathcal{C}=L\,\}
\end{equation}
\noindent We need to characterize the subspace $S_L$. To this end
notice that $S_{L}$ satisfies

\begin{equation} \label{clas}
S_{L}\cap\mathcal{C}=L,
\end{equation}

\noindent (notice that $S_{L}$ may be infinite dimensional). Let
consider the class $[S]$ of elements satisfying (\ref{clas}):
$\,\,[S]=L$, with $S\in G(\mathcal{A})$ being an element of the
class.

\vskip 3mm \nd We introduce some math-notation now. First, closure
of a set will be indicated by an overbar over the set's name.
Secondly, given a set $M$ we will denote by $<M>_E$ the set of
linear combinations of $M-$elements with coefficients extracted
from the set of scalars $E$.  Then

\begin{equation}
S\cap\mathcal{C}\subseteq\overline{<S\cap\mathcal{C}>}_{\mathcal{R}}\subseteq
S\Rightarrow
S\cap\mathcal{C}\cap\mathcal{C}\subseteq\overline{<S\cap\mathcal{C}>}_{\mathcal{R}}\cap\mathcal{C}
\subseteq S\cap\mathcal{C}\Rightarrow
\end{equation}

\begin{equation}
\overline{<S\cap\mathcal{C}>}\cap\mathcal{C}=S\cap\mathcal{C}
\end{equation}

\noindent Accordingly,  $\overline{<S\cap\mathcal{C}>}$ and $S$
are in the same class $L$. Note that
$\overline{<S\cap\mathcal{C}>}\subseteq S$ and if $S$ equals
$S_{L}$, we will also have  $S_{L}\subseteq
\overline{<S_{L}\cap\mathcal{C}>}$ (because  $S_{L}$ is the
intersection of all the elements in the class). Consequently,

\begin{equation}
\overline{<S_{L}\cap\mathcal{C}>}=S_{L}
\end{equation}
The above equation also implies that $S_{L}$ is the unique
subspace with that property, because if we choose $S'$ such that
$S'\cap\mathcal{C}=S\cap\mathcal{C}$, then

\begin{equation}
S=\overline{<S\cap\mathcal{C}>}=\overline{<S'\cap\mathcal{C}>}=S'
\end{equation}
\noindent Summing up, the representative of a class $L$ is the
unique $\mathcal{R}$-subspace $S\subseteq\mathcal{A}$ such that

\begin{equation}
S=\overline{<S\cap\mathcal{C}>}_{\mathcal{R}}
\end{equation}
\noindent and we have proved that it is equal to $S_{L}$, that we
call the {\it good representative}.

\subsubsection{Math interlude}

\nd We need at this point to remind the reader of some lattice
concepts (See appendix \ref{s:ApendixB}). Let $\mathcal{P}$ be a poset, partially
ordered by the order relation ``less or equal". An element $a \in
\mathcal{P}$ is called an {\sf atom} if it covers some minimal
element of  $\mathcal{P}$. As a result, an atom is never minimal. A
poset $\mathcal{P}$ is called atomic if for every element $p \in
\mathcal{P}$ (that is not minimal) an atom $a$ exist such that $a
\le p$.
 For instance,\newline
\nd 1. Let $A$ be a set and  $\mathcal{P} =2^A$ its power set.
$\mathcal{P}$  is a poset ordered by  the ``inclusion" relation,
with a unique minimal element. Thus, all singleton subsets $a$ of
$A$ are atoms in  $\mathcal{P}$ (a set is a singleton if and only if
its cardinality is 1). Accordingly, in the set-theoretic
construction of the natural numbers, the number 1 is defined as the
singleton $\{0\}$.
\newline \nd 2. The set of positive integeres is partially ordered if we define $a\le b$ to mean that
$b/a$ is a positive integer. Then 1 is a minimal element and any
prime number $p$ is an atom. \vskip 3mm \nd Given a lattice
$\mathcal{L}$ with underlying poset  $\mathcal{P}$, an element $a
\in  \mathcal{L}$  is called an atom (of  $\mathcal{L}$) if it is an
atom in  $\mathcal{P}$. A lattice is called an atomic lattice if its
underlying poset is atomic. An atomistic lattice is an atomic
lattice such that each element that is not minimal is a join of
atoms. \vskip 3mm \nd Finally (see Appendix \ref{s:ApendixB}), we call   {\it
modular} a  lattice  that satisfies the following self-dual
condition (modular law) $ x \leq b$ implies $x \vee (a \wedge b) =
(x \vee a) \wedge b$, where $\le$ is the partial order, and $\vee$
and $\wedge$ (join and meet, respectively) are the operations of the
lattice.

\subsection{Operations in $\mathcal{L}$}
Let us now define ``$\vee$", ``$\wedge$" and ``$\neg$" operations
and a partial ordering relation ``$\longrightarrow$" (or
equivalently ``$\leq$") in $\mathcal{L}$ as

\begin{enumerate}

\item
\begin{equation}
(S\cap\mathcal{C}) \wedge (T\cap\mathcal{C})\Longleftrightarrow
(\overline{<S\cap\mathcal{C}>}\cap\overline{<T\cap\mathcal{C}>})\cap\mathcal{C}
\end{equation}

\item
\begin{equation}
(S\cap\mathcal{C}) \vee (T\cap\mathcal{C})\Longleftrightarrow
(\overline{<S\cap\mathcal{C}>}+\overline{<T\cap\mathcal{C}>})\cap\mathcal{C}
\end{equation}

\item
\begin{equation}
(S\cap\mathcal{C}) \longrightarrow
(T\cap\mathcal{C})\Longleftrightarrow
(S\cap\mathcal{C})\subseteq(T\cap\mathcal{C})
\end{equation}

\item
\begin{equation}
\neg(S\cap\mathcal{C})\Longleftrightarrow
<S\cap\mathcal{C}>^\perp\cap\mathcal{C}. \label{fouroper}
\end{equation}
\end{enumerate}

\nd Let us consider now some consequences of the above
definitions. It is easy to see that  (see Appendix \ref{s:ApendixB}):

\begin{prop}
\nd $\mathcal{L}$ is an atomic and complete lattice. If
$\dim(\mathcal{H})<\infty$, $\mathcal{L}$ is a modular lattice.
\end{prop}
\nd $\mathcal{L}$ is not an orthocomplemented lattice, but it is
easy to show that non-contradiction holds

\begin{equation}
L\wedge\neg L=\mathbf{0}
\end{equation}

\noindent and also contraposition

\begin{equation}
\noindent L_{1}\leq L_{2}\Longrightarrow \neg L_{2}\leq \neg L_{1}
\end{equation}
\nd The following proposition is important because it {\it links
atoms and quantum states}:

\begin{prop}
There is a one to one correspondence between the states of the
system and the atoms of $\mathcal{L}$.
\end{prop}

\begin{proof}
Let $\rho\in\mathcal{C}$ be any state. Next, consider the subspace
$S_{\rho}=<\rho>$. Then, $Q_{\rho}=S_{\rho}\cap\mathcal{C}=\rho$ is a singleton because
$\rho$ is the only {\sf trace one} operator in $S_{\rho}$, and thus an  atom since it covers the minimal element  $\rho\in\mathcal{C}$. This
completes the bijective correspondence.
\end{proof}

\noindent The reader is now adviced to peruse Appendix \ref{s:ApendixC} and remember that a {\it face} of a convex set $\mathcal{C}$ is the intersection of $\mathcal{C}$
 with a supporting hyperplane of $\mathcal{C}$. It is well known \cite{bengtssonyczkowski2006} that in
the finite dimensional case there is a lattice isomorphism between
the complemented and complete lattice of faces of the convex set
$\mathcal{C}$ and $\mathcal{L}_{v\mathcal{N}}$ (see Appendix
\ref{s:ApendixB} for definitions). We also repeat that
$\mathcal{L}$ is a set of intersections, Which ones?  Those
between $\mathcal{R}-$subspaces that belong to $G(\mathcal{A})$
and $\mathcal{C}$. Let us put forward the following proposition

\begin{prop}\label{p:FacesAreElements}
Every face of $\mathcal{C}$ is an element of $\mathcal{L}$.
\end{prop}

\begin{proof}
If $F$ is a face of $\mathcal{C}$, then there exists a closed
hyperplane $H_{F}\subseteq\mathcal{A}$ such that
$H_{F}\cap\mathcal{C}=F$ and $\mathcal{C}$ stands on one side of the
half-spaces defined by $H_{F}$. If $\hat{H}_{F}$ is the continuous
linear functional on $\mathcal{A}$ defined by the real sub-vector
space $H_{F}$ (that is, the map that assigns a scalar to each member
of  $H_{F}$), then there exists $\alpha\in\mathbb{R}$ such that
$$H_{F}=\{x\,|\,\hat{H}_{F}(x)=\alpha\},\quad
\mathcal{C}\subseteq \{x\,|\,\hat{H}_{F}(x)\leq\alpha\}.$$

\noindent Consider the continuous linear functional
\begin{equation}
\hat{L}(x):=\hat{H}_{F}(x)-\alpha\text{tr}(x).
\end{equation}
\noindent In the Banach space of trace class operators
$\mathcal{B}_1$ we have the norm $\| T\|=\text{tr}(|T|)$, then the
linear functional $\text{tr}:\mathcal{B}_1\rightarrow\mathbb{R}$ is
continuous. Given that $\hat{H}_{F}$ is continuous on $\mathcal{A}$
and $\mathcal{B}_1\subseteq \mathcal{A}$ we have that $\hat{L}$ is
continuous on $\mathcal{B}_1$. Its kernel, or null space (the set of
all functions that the functional maps to zero) will be a closed
subspace of $\mathcal{B}_1$ (we call it $S$). Now notice that
$\mathcal{C}\subseteq \mathcal{B}_1$ and

\begin{eqnarray}
&S\cap\mathcal{C}=\{x\,|\,\hat{H}_{F}(x)=\alpha\mbox{tr}(x)\}\cap\mathcal{C}=\{x\,|\,\hat{H}_{F}(x)=\alpha\mbox{tr}(x)\,\mbox{and}\,\mbox{tr}(x)=1\}\cap\mathcal{C}=&\nonumber\\
&\{x\,|\,\hat{H}_{F}(x)=\alpha\}\cap\mathcal{C}=H_{F}\cap\mathcal{C}=F.&
\end{eqnarray}

\noindent We conclude that $F\in\mathcal{L}$.

\end{proof}

\noindent Using this result and the isomorphism between faces of the
convex set and subspaces, we conclude that (at least for the finite
dimensional case)

\begin{coro}
The complete lattice of faces of the convex set $\mathcal{C}$ is
essentially a subposet of $\mathcal{L}$.
\end{coro}

\noindent The previous Corollary shows that $\mathcal{L}$ and
$\mathcal{L}_{v\mathcal{N}}$ are closely connected. We have stated
that any convex subset of $\mathcal{C}$ can be considered as a
probability space by itself, and certainly, the elements of
$\mathcal{L}$ are convex sets, because they are built as the
intersection of a closed subspace of $\mathcal{A}$ and a convex set
($\mathcal{C}$). Thus, we have found not only that closed subspaces
of $\mathcal{H}$ may be considered as yes-no tests (via their one to
one correspondence with projection operators), but also that they
may be considered as probability spaces (endowed with ``mixing
principle" mentioned above), because of their one to one
correspondence with the elements of $\mathcal{L}$. \vskip 3mm
\noindent The crucial question now is: what is the relationship
between their respective operations? If $F_{1}$ and $F_{2}$ are
faces we have

\begin{enumerate}

\item[($\wedge$)]
$F_1,F_2\in\mathcal{L}_{v\mathcal{N}}$, then $F_1\wedge F_2$ in
$\mathcal{L}_{v\mathcal{N}}$ is the same as in $\mathcal{L}$.
Thus, the inclusion
$\mathcal{L}_{v\mathcal{N}}\subseteq\mathcal{L}$ preserves the
$\wedge$-operation.

\item[($\vee$)]
$F_1\vee_{\mathcal{L}} F_2\leq
F_1\vee_{\mathcal{L}_{v\mathcal{N}}} F_2$ and $F_1\leq
F_2\Rightarrow F_1\vee_{\mathcal{L}}
F_2=F_1\vee_{\mathcal{L}_{v\mathcal{N}}} F_2=F_2$

\item[$(\neg$)]
$\neg_{\mathcal{L}}F\leq \neg_{\mathcal{L}_{v\mathcal{N}}}F$
\end{enumerate}

\nd  Given two systems with Hilbert spaces $\mathcal{H}_{1}$ and
$\mathcal{H}_{2}$, we can construct the lattices $\mathcal{L}_{1}$
and $\mathcal{L}_{2}$. We can also built up $\mathcal{L}$, the
lattice associated to the product space
$\mathcal{H}_{1}\otimes\mathcal{H}_{2}$. We define

\begin{equation}
\Psi:\mathcal{L}_{1}\times\mathcal{L}_{2}\longrightarrow\mathcal{L}\quad|\quad(S_{1}\cap\mathcal{C}_{1},S_{2}\cap\mathcal{C}_{2})\longrightarrow
S\cap\mathcal{C}
\end{equation}
\noindent where
$S=(\overline{<S_{1}\cap\mathcal{C}_{1}>}\otimes\overline{<S_{2}\cap\mathcal{C}_{2}>})$.
In terms of good representatives, $\Psi([S_1],[S_2])=[S_1\otimes
S_2]$. An equivalent definition (in the finite dimensional case)
states that $\Psi$ is the induced morphism in the quotient
lattices (See Appendix \ref{s:ApendixA}) of the tensor map:

\begin{equation}
G(\mathcal{A}_1)\times G(\mathcal{A}_2)\rightarrow
G(\mathcal{A}_1\otimes_{\mathcal{R}} \mathcal{A}_2)\cong
G(\mathcal{A})
\end{equation}

\noindent Given $L_{1}\in\mathcal{L}_{1}$ and
$L_{2}\in\mathcal{L}_{2}$, we can define the following convex tensor
product:

\begin{definition}\label{d:convex tensor product}
$L_{1}\widetilde{\otimes}\,L_{2}:=\{\sum\lambda_{ij}\rho_{i}^{1}\otimes\rho_{j}^{2}\,|\,\rho_{i}^{1}
\in L_{1},\,\,\rho_{j}^{2}\in L_{2},\,\, \sum\lambda_{ij}=1
\,\,\mbox{and} \,\,\lambda_{ij}\geq 0\}$
\end{definition}

\noindent This product is formed by all possible convex
combinations of tensor products of elements of $L_{1}$ and
elements of $L_{2}$, and it is again a convex set. Let us compute
$\mathcal{C}_{1}\widetilde{\otimes}\,\mathcal{C}_{2}$. Remember
that $\mathcal{C}_{1}=[\mathcal{A}_{1}]\in\mathcal{L}_{1}$ and
$\mathcal{C}_{2}=[\mathcal{A}_{2}]\in\mathcal{L}_{2}$:

\begin{equation}
\mathcal{C}_{1}\widetilde{\otimes}\,\mathcal{C}_{2}=
\{\sum\lambda_{ij}\rho_{i}^{1}\otimes\rho_{j}^{2}\,|\,\rho_{i}^{1}
\in \mathcal{C}_{1},\,\,\rho_{j}^{2}\in \mathcal{C}_{2},\,\,
\sum\lambda_{ij}=1 \,\,\mbox{and} \,\,\lambda_{ij}\geq 0\}.
\end{equation}
\noindent Thus, if $\mathcal{S(\mathcal{H})}$ is the set of all
separable states, we have by definition

\begin{equation}\label{e:separablestates}
\mathcal{S}(\mathcal{H})=\overline{\mathcal{C}_{1}\widetilde{\otimes}\,\mathcal{C}_{2}}
\end{equation}
If the whole system is in a state $\rho$, using partial traces we
can define states for the subsystems $\rho_{1}=tr_{1}(\rho)$ and a
similar definition for $\rho_{2}$. Then, we can consider the maps

\begin{equation}
\mbox{tr}_{i}:\mathcal{C}\longrightarrow \mathcal{C}_{j}
\quad|\quad \rho\longrightarrow \mbox{tr}_{i}(\rho)
\end{equation}
\noindent from which we can construct the induced projections:

\begin{equation}
\tau_{i}:\mathcal{L}\longrightarrow \mathcal{L}_{i} \quad|\quad
S\cap\mathcal{C}\longrightarrow \mbox{tr}_{i}(
\overline{<S\cap\mathcal{C}>} )\cap \mathcal{C}_i=
\mbox{tr}_{i}( \overline{<S\cap\mathcal{C}>} )\cap \mathcal{C}_i
\end{equation}
In terms of good representatives $\tau_i([S])=[\mbox{tr}_i(S)]$.
As a consequence,  we can define the product map

\begin{equation}
\tau:\mathcal{L}\longrightarrow\mathcal{L}_{1}\times\mathcal{L}_{2}
\quad|\quad L\longrightarrow(\tau_{1}(L),\tau_{2}(L))
\end{equation}
\noindent The maps defined in this section are illustrated in
Figure \ref{f:maps1}.

\begin{figure}\label{f:maps1}
\begin{center}
\unitlength=1mm
\begin{picture}(5,5)(0,0)
\put(-3,23){\vector(-1,-1){20}} \put(3,23){\vector(1,-1){20}}
\put(-2,4){\vector(0,2){16}} \put(2,20){\vector(0,-2){16}}
\put(8,0){\vector(3,0){15}} \put(-8,0){\vector(-3,0){15}}
\put(0,25){\makebox(0,0){$\mathcal{L}$}}
\put(-27,0){\makebox(0,0){${\mathcal{L}_{1}}$}}
\put(27,0){\makebox(0,0){${\mathcal{L}_{2}}$}}
\put(0,0){\makebox(0,0){${\mathcal{L}_{1}\times\mathcal{L}_{2}}$}}
\put(-1,11){\makebox(-10,0){$\psi$}}
\put(-1,11){\makebox(13,0){$\tau$}}
\put(-15,16){\makebox(0,0){$\tau_{1}$}}
\put(15,16){\makebox(0,0){$\tau_{2}$}}
\put(-15,2){\makebox(0,0){$\pi_{1}$}}
\put(15,2){\makebox(0,0){$\pi_{2}$}}
\end{picture}
\caption{The different maps between $\mathcal{L}_{1}$,
$\mathcal{L}_{2}$, ${\mathcal{L}_{1}\times\mathcal{L}_{2}}$, and
$\mathcal{L}$. $\pi_{1}$ and $\pi_{2}$ are the canonical
projections.}
\end{center}
\end{figure}
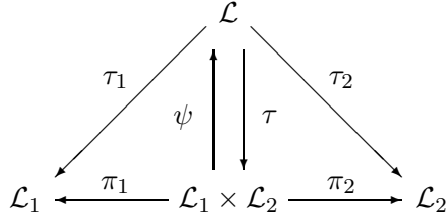

\section{The Lattice of Convex Subsets}\label{s:Convex lattice}

The elements of $\mathcal{L}$ are the intersections between closed
(real) subspaces of $\mathcal{A}$ and $\mathcal{C}$. Now we drop
the restriction and consider {\it all} possible probability
subspaces of $\mathcal{C}$, i.e., all of their possible convex
subsets. Because of linearity, partial trace operators preserve
convexity and so they will map probability spaces of the system
into probability spaces of the subsystems, as desired. Let us
begin  then by considering the set of all convex subsets of
$\mathcal{C}$ (\underline{not only closed ones}):

\begin{definition}
\noindent $\mathcal{L}_{\mathcal{C}}:=\{C\subseteq\mathcal{C}\,|\,
\mbox{C is a convex subset of} \,\,\,\mathcal{C}\}$
\end{definition}

\nd In order to give $\mathcal{L}_{\mathcal{C}}$ a lattice
structure, we introduce the following operations (where $conv(A)$
stands for convex hull of a given set $A$):

\begin{definition}\label{definitionlattice}

For all $C,C_1,C_2\in\mathcal{L}_{\mathcal{C}}$
\begin{enumerate}

\item[$\wedge$]:\,\,
$C_1\wedge C_2:= C_1\cap C_2$

\item[$\vee$]:\,\,
$C_1\vee C_2:=conv(C_1,C_2)$. It is again a convex set, and it is
included in $\mathcal{C}$ (using convexity).

\item[$\neg$]:\,\,
$\neg C:=C^{\perp}\cap\mathcal{C}$

\item[$\longrightarrow$]:\,\,
$C_1\longrightarrow C_2:= C_1\subseteq C_2$

\end{enumerate}

\end{definition}
\noindent With the operations of definition
\ref{definitionlattice}, it is apparent that
$(\mathcal{L}_{\mathcal{C}};\longrightarrow)$ is a poset. If we
set $\emptyset=\mathbf{0}$ and $\mathcal{C}=\mathbf{1}$, then,
$(\mathcal{L}_{\mathcal{C}};\longrightarrow;\mathbf{0};\emptyset=\mathbf{0})$
will be a bounded poset.

\nd It is very easy to show that

\begin{prop}
$(\mathcal{L}_{\mathcal{C}};\longrightarrow;\wedge;\vee)$
satisfies

\begin{enumerate}

\item[$(a)$]
$C_1\wedge C_1= C_1$

\item[$(b)$]
$C_1\wedge C_2=C_2\wedge C_1$

\item[$(c)$]
$C_1\vee C_2=C_2\vee C_1$

\item[$(d)$]
$C_1\wedge (C_2\wedge C_3)=(C_1\wedge C_2)\wedge C_3$

\item[$(e)$]
$C_1\vee (C_2\vee C_3)=(C_1\vee C_2)\vee C_3$

\item[$(f)$]
$C_1\wedge (C_1\vee C_2)=C_1$

\item[$(g)$]
$C_1\vee (C_1\wedge C_2)=C_1$
\end{enumerate}

\end{prop}

\nd Regarding the ``$\neg$" operation, if $C_1\subseteq  C_2$,
then $C_2^{\perp}\subseteq  C_1^{\perp}$. Accordingly,
$C_2^{\perp}\cap \mathcal{C}\subseteq C_1^{\perp}\cap\mathcal{C}$,
and hence

\begin{equation}
C_1\longrightarrow  C_2\Longrightarrow \neg C_2\longrightarrow
\neg C_1
\end{equation}
\noindent Given that $C\cap(C^{\perp}\cap\mathcal{C})=\emptyset$,
we also have:

\begin{equation}
C\wedge(\neg C)=\mathbf{0}
\end{equation}
\noindent Contraposition and non contradiction thus hold. But if we
take the proposition $C=\{\frac{1}{N}\mathrm{1}\}$, then an easy
calculation yields $\neg C=\mathbf{0}$. Then, $\neg(\neg
C)=\mathbf{1}$, and thus $\neg(\neg C)\neq C$ in general. Double
negation does not hold. Consequently, $\mathcal{L}_{\mathcal{C}}$ is
not an ortholattice. $\mathcal{L}_{\mathcal{C}}$ is a lattice which
includes all convex subsets of the quantum space of states. It
includes $\mathcal{L}$, and then all quantum states (including all
improper mixtures), {\it as propositions} (understood as elements of
the lattice). Compare with classical physics, where the lattice of
propositions is formed by all measurable subsets of Gibbs'
phase-space (the space of states).

\noindent As any convex subset of $\mathcal{C}$  is a probability
space by itself, we can define on each of these subsets a lattice
structure in  analogous way to that  used for
$\mathcal{L}_{\mathcal{C}}$. In other words, we encounter an
inheritance of the same structure.

\subsection{Effects, mean values and maximum entropy
principle}\label{s:Max-Ent}

\nd Let us discuss, as an example, the family of elements of
$\mathcal{L}_{\mathcal{C}}$ that is associated to effects. Given
an effect $E$, consider the set of states

\begin{equation}
C_{(E,\lambda)}:=\{\rho\in \mathcal{C}\,|\,\mbox{tr}(\rho
E)=\lambda,\,\,\lambda\in[0,1]\}.
\end{equation}
\noindent It is easy to verify that $C_{E}$ is a convex set and,
consequently,  an element of $\mathcal{L}_{\mathcal{C}}$.
$C_{(E,\lambda)}$ represents all the states for which the
probability of having the effect $E$ is equal to $\lambda$.
Furthermore, there exists $S$, an $\mathcal{R}$-subspace of
$\mathcal{A}$, such that

\begin{equation}\label{e:effectconvex}
C_{(E,\lambda)}=S\cap\mathcal{C},
\end{equation}
\noindent and thus, $C_{E}$ is also an element of $\mathcal{L}$. The
proof of (\ref{e:effectconvex}) is as follows. Consider the linear
functional

\begin{equation}
F_{(E,\lambda)}(\rho):=\mbox{tr}(E\rho)-\lambda\mbox{tr}(\rho),
\end{equation}

\noindent As in the proof of Proposition \ref{p:FacesAreElements},
we will have that the set $S=Ker(F_{(E,\lambda)})$ is a closed
subspace of $\mathcal{B}_1$. Then, an element $\rho\in
C_{(E,\lambda)}$ will be given by an element of $\mathcal{C}$ of
trace one, and thus
$F_{(E,\lambda)}(\rho)=\mbox{tr}(E\rho)-\lambda$, plus the equality
to $\lambda$ requirement imply that $\rho$ also belongs to $S$. We
have thus demonstrated Eq. (\ref{e:effectconvex}). More generally,
if we have the equation for the mean value of an operator

\begin{equation}\label{e:setofmatrixes}
\langle R\rangle=r,
\end{equation}

\noindent the operator may be considered as represented by the set
of density matrices which satisfy the above equality. The ensuing
set is obtained as the intersection of the Kernel of the functional
$F_{R}(\rho):=\mbox{tr}(R\rho)-r\mbox{tr}(\rho)$ with $\mathcal{C}$.
Accordingly, each equation of the form (\ref{e:setofmatrixes})
(understood as an equation to be solved) can be represented as an
element $C\in\mathcal{L}$, and also as an element of
$\mathcal{L}_{\mathcal{C}}$. Note that $\mathcal{C}$ is also a
closed convex set because it is the intersection of the following
closed half-spaces
$$H_x=\{\rho\in\mathcal{A}\,|\,x^\dag\rho x=0\},\quad
\mathcal{C}=\bigcap_{x\in\mathcal{H}}\{\rho \,|\,x^\dag\rho
x\geq0\}\cap\{\rho\,|\,\text{tr}(\rho)=1\}.$$

\noindent With such materials at hand, we can now re-express the
celebrated Jaynes' maximum entropy principle
\cite{Jaynes-1957a,Jaynes-1957b} in  lattice theoretical fashion.
Jaynes assumes that for a quantum system one knows the mean values

\begin{eqnarray}\label{e:conditionsmean}
&\langle R_{1}\rangle=r_{1}&\nonumber\\
&\langle R_{2}\rangle=r_{2}&\nonumber\\
&\vdots&\nonumber\\
&\langle R_{n}\rangle=r_{n},
\end{eqnarray}

\noindent and we want to determine the most unbiased density matrix
compatible with conditions (\ref{e:conditionsmean}). Jaynes' MaxEnt
principle asserts that the solution to this problem is given by the
density matrix that maximizes entropy, given by

\begin{equation}
\rho_{max-ent}=\exp^{-\lambda_{0}\mathbf{1}-\lambda_{1}R_{1}-\cdots-\lambda_{n}R_{n}},
\end{equation}
\noindent where the $\lambda$'s are Lagrange multipliers
satisfying

\begin{equation}
r_{i}=-\frac{\partial}{\partial\lambda_{i}}\ln Z,
\end{equation}
\noindent while the partition function reads

\begin{equation}
Z(\lambda_{1}\cdots\lambda_{n})=\mbox{tr}[\exp^{-\lambda_{1}R_{1}-\cdots-\lambda_{n}R_{n}}],
\end{equation}
\noindent and the normalization condition is

\begin{equation}
\lambda_{0}=\ln Z.
\end{equation}
Our point here is that the set of conditions
(\ref{e:conditionsmean}) can be expressed in an explicit lattice
theoretical form as follows. Using a similar procedure as in
(\ref{e:setofmatrixes}) we easily conclude that each of the
equations in (\ref{e:conditionsmean}) can be represented as a convex
(and closed) sets $C_{R_{i}}$. In this way we can now express
conditions (\ref{e:conditionsmean}) via the lattice theoretical
expression

\begin{equation}
C_{max-ent}:=\bigcap_{i}C_{R_{i}}=\bigwedge_{i}C_{R_{i}}.
\end{equation}

\nd Now, $C_{max-ent}$ is also an element of
$\mathcal{L}_{\mathcal{C}}$ (but not necessarily of $\mathcal{L}$)
and we must maximize entropy in it. We have thus encountered a
MaxEnt-lattice theoretical expression: {\it given a set of
conditions represented generally by convex subsets $C_{i}$, one
should maximize the entropy in the set
$C_{max-ent}=\bigwedge_{i}C_{i}$.}

\subsection{The Relationship Between $\mathcal{L}_{v\mathcal{N}}$, $\mathcal{L}$ and
$\mathcal{L}_{\mathcal{C}}$}\label{s:The Relationship}

\nd Please, refer here to Appendix \ref{s:ApendixB}.

\begin{prop}\label{p:Inclusion of Lvn}
For finite dimension
$\mathcal{L}_{v\mathcal{N}}\subseteq\mathcal{L}\subseteq\mathcal{L}_{\mathcal{C}}$
as posets.
\end{prop}

\begin{proof}
We have already seen that
$\mathcal{L}_{v\mathcal{N}}\subseteq\mathcal{L}$ as sets. Moreover
it is easy to see that if $F_1\leq F_2$ in
$\mathcal{L}_{v\mathcal{N}}$ then $F_1\leq F_2$ in $\mathcal{L}$.
This is so because both orders are set theoretical inclusions.
Similarly, if $L_{1},L_{2}\in\mathcal{L}$, because intersection of
convex sets yields a convex set (and closed subspaces are convex
sets also), $L_{1},L_{2}\in\mathcal{L}_{\mathcal{C}}$, then we
obtain set theoretical inclusion. And, again, because of both
orders are set theoretical inclusions, we obtain that they are
included as posets.
\end{proof}

\nd Regarding the $\vee$ operation, let us compare
$\vee_{\mathcal{L}_{v\mathcal{N}}}$, $\vee_{\mathcal{L}}$ and
$\vee_{\mathcal{L}_{C}}$. If $L_{1},L_{2}\in\mathcal{L}$, then they
are convex sets and so, $L_{1},L_{2}\in\mathcal{L}_{\mathcal{C}}$.
Thus we can compute

\begin{equation}
L_{1}\vee_{\mathcal{L}_{C}}L_{2}=conv(L_{1},L_{2}).
\end{equation}

\noindent On the other hand (if $S_{1}$ and $S_{2}$ are good
representatives for $L_{1}$ and $L_{2}$), then:

\begin{equation}
L_{1}\vee_{\mathcal{L}}L_{2}=(\overline{<S_{1}\cap\mathcal{C}>}+\overline{<S_{2}\cap\mathcal{C}>})\cap\mathcal{C}.
\end{equation}

\noindent The direct sum of the subspaces
$\overline{<S_{1}\cap\mathcal{C}>}$ and
$\overline{<S_{2}\cap\mathcal{C}>}$ contains as a particular case
all convex combinations of elements of $L_{1}$ and $L_{2}$. We then
conclude (including infinite dimensional case)

\begin{equation}
L_{1}\vee_{\mathcal{L}_{C}}L_{2}\leq L_{1}\vee_{\mathcal{L}}L_{2}.
\end{equation}

\nd The faces of $\mathcal{C}$ can be considered as elements of
$\mathcal{L}_{C}$ because they are convex. If $F_{1}$ and $F_{2}$
are faces we can also state (finite dimension)

\begin{equation}
F_{1}\vee_{\mathcal{L}_{C}}F_{2}\leq
F_{1}\vee_{\mathcal{L}}F_{2}\leq
F_{1}\vee_{\mathcal{L}_{v\mathcal{N}}}F_{2}.
\end{equation}

\nd Intersection of convex sets is the same as intersection of
elements of $\mathcal{L}$ and so we have (infinite dimension)

\begin{equation}
L_{1}\wedge_{\mathcal{L}_{C}}L_{2}=
L_{1}\wedge_{\mathcal{L}}L_{2},
\end{equation}

\noindent and similarly (finite dimension)

\begin{equation}
F_{1}\wedge_{\mathcal{L}_{v\mathcal{N}}}F_{2}=
F_{1}\wedge_{\mathcal{L}_{C}}F_{2}= F_{1}\wedge_{\mathcal{L}}F_{2}.
\end{equation}

\nd What is the relationship between
$\neg_{\mathcal{L}_{\mathcal{C}}}$ and $\neg_{\mathcal{L}}$?
Suppose that $L_{1}\in\mathcal{L}$, then they are convex sets as
well and  $L_{1}\in\mathcal{L}_{\mathcal{C}}$. We can now compute
$\neg_{\mathcal{L}_{\mathcal{C}}}L_{1}$ and obtain

\begin{equation}
\neg_{\mathcal{L}_{\mathcal{C}}}L_{1}=L_{1}^{\perp}\cap\mathcal{C}.
\end{equation}

\noindent On the other hand, if $L_{1}=S\cap\mathcal{C}$, with $S$
a good representative,

\begin{equation}
\neg_{\mathcal{L}}L_{1}=<S\cap\mathcal{C}>^{\perp}\cap\mathcal{C}.
\end{equation}

\noindent As $L_{1}\subseteq <S\cap\mathcal{C}>$, then
$<S\cap\mathcal{C}>^{\perp}\subseteq L_{1}^{\perp}$, and finally
(infinite dimension)

\begin{equation}
\neg_{\mathcal{L}}L_{1}\leq\neg_{\mathcal{L}_{\mathcal{C}}}L_{1}.
\end{equation}

\section{Identifying entanglement in the finite dimensional
case}\label{s:EntanglementWittness}

\noindent After introducing the new lattice theoretical frameworks,
it is interesting to look back again to the finite dimensional case,
condition that will we assume through this section. In doing so, we
will find a separability criteria. We will include in this section
proofs of various well established facts because in doing
so, we illuminate the lattice structure behind those facts which
allow us to deduce our novel separability criteria.
In future works we will use these proofs to compute effectively
\emph{the volume of the space of separable states}. Let
$\mathcal{S}_{0}(\mathcal{H})$ denote the space of product states

\begin{equation}
\mathcal{S}_{0}(\mathcal{H})=\{a\otimes
b\,|\,a\in\mathcal{C}_1,\,b\in\mathcal{C}_2\} \subseteq
S(\mathcal{H})
\end{equation}

\noindent By definition, $S(\mathcal{H})$ is the convex hull of
$\mathcal{S}_0(\mathcal{H})$. Using a well known fact about the
theory of convex sets, we know that $S(\mathcal{H})$ is the
intersection of all the closed half-hyperplanes containing
$\mathcal{S}_{0}(\mathcal{H})$ \cite[11.5.1]{rocka}. For any
functional $\ell$ representing a half-plane and a real number $m$,
we have

\begin{equation}
\mathcal{S}_{0}(\mathcal{H})\subseteq\{\ell\geq m\}\iff
\ell(\mathcal{S}_{0}(\mathcal{H}))\subseteq [m,+\infty)\iff
\ell(\mathcal{S}_{0}(\mathcal{H}))\text{ has a minimum value }m.
\end{equation}

\begin{rema}
We  work with i) $\mathcal{B}(\mathcal{H})$ and ii) its isomorphic
space $\mathcal{B}(\mathcal{H}_1)\otimes
\mathcal{B}(\mathcal{H}_2)$. The composition for
$\rho_1,\rho_2\in\mathcal{B}(\mathcal{H})$ is $\rho_1\rho_2,$ and
for $a\otimes b,c\otimes d\in\mathcal{B}(\mathcal{H}_1)\otimes
\mathcal{B}(\mathcal{H}_2),$ is $ac\otimes bd$. The induced trace
in $\mathcal{B}(\mathcal{H}_1)\otimes \mathcal{B}(\mathcal{H}_2)$
is $\text{tr}(a\otimes b)=\text{tr}(a)\text{tr}(b)$. In
$\mathcal{B}(\mathcal{H})$ we have a canonical inner product,
$$\langle \rho_1,\rho_2\rangle=\text{tr}(\rho_1\rho_2^\dag)$$
In $\mathcal{B}(\mathcal{H}_1)\otimes \mathcal{B}(\mathcal{H}_2)$
we have the induced inner product,
$$\langle a\otimes b,c\otimes d\rangle=\text{tr}(ac^\dag)\text{tr}(bd^\dag)$$
In particular, in
$\mathcal{A}(\mathcal{H}_1)\otimes\mathcal{A}(\mathcal{H}_2)$ one
has

\begin{equation}
\langle a\otimes b,c\otimes d\rangle=\text{tr}(ac)\text{tr}(bd).
\end{equation}
\nd This inner product gives to $\mathcal{A}(\mathcal{H})$ the
structure of an $\mathcal{R}$-vector space with an inner product.
Hence, any linear functional $\ell$ is associated to a vector
$\rho\in\mathcal{A}(\mathcal{H})$ such that
$\ell=\langle\rho,-\rangle$.
\end{rema}

\noindent It is important to notice here that any functional of
the form $\ell=\langle\rho,-\rangle$ defines by varying
$m\in\mathcal{R}$ a family of elements of $\mathcal{L}_{C}$ and
$\mathcal{L}$ by considering the families of convex subsets
$\{\ell\leq m\}\cap\mathcal{C}$ and $\{\ell= m\}\cap\mathcal{C}$
respectively. It is possible to show that

\begin{prop}
The space of product states $\mathcal{S}_{0}(\mathcal{H})$ is
compact. In particular, every linear functional of the form
$\ell_\rho$ on $\mathcal{S}_{0}(\mathcal{H})$ has a maximum $M_\rho$
and a minimum $m_\rho$ values.
\end{prop}

\begin{proof}
It is well known that for every matrix in $a\in\mathcal{C}_1$ there
exist a unitary matrix $U\in\mathcal{U}(\mathcal{H}_1)$ such that
$UaU^\dag$ is diagonal. In particular, if we take the subgroup
$$\mathcal{U}(\mathcal{H}_1)\times\mathcal{U}(\mathcal{H}_2)\hookrightarrow
\mathcal{U}(\mathcal{H})=\{U\in
\mathcal{B}(\mathcal{H})\,|\,UU^\dag=I\}$$ we have that every
product state is conjugate to a product of two diagonal matrixes
$$(U_1,U_2)(a_1\otimes a_2)(U_1^\dag,U_2^\dag)=(U_1a_1U_1^\dag)\otimes (U_2a_2U_2^\dag)=
d_1\otimes d_2.$$ Note that the subgroup
$\mathcal{U}(\mathcal{H}_1)\times\mathcal{U}(\mathcal{H}_2)$
preserves the space of product states and the space of separable
states. Let us define the space of diagonal product states
$$\mathcal{D}=\{d_1\otimes d_2\in\mathcal{S}\,|\, d_1=\text{diag}(\alpha_1,\ldots,\alpha_n),\,
d_2=\text{diag}(\beta_1,\ldots,\beta_m)\}.$$ Note that $\mathcal{D}$
is isomorphic, as a topological space, to the product of simplexes
$\triangle^n\times\triangle^m$. Indeed, we have the following
continuous surjective map from a compact space to
$\mathcal{S}_{0}(\mathcal{H})$

\begin{equation}
\Gamma:
\mathcal{U}(\mathcal{H}_1)\times\mathcal{U}(\mathcal{H}_2)\times\triangle^n\times\triangle^m
\longrightarrow\mathcal{S}_{0}(\mathcal{H}),
\end{equation}

\begin{equation}
\Gamma(U_1,U_2,(\alpha_1,\ldots,\alpha_n),(\beta_1,\ldots,\beta_m))=
U_1^\dag\text{diag}(\alpha_1,\ldots,\alpha_n)U_1\otimes
U_2^\dag\text{diag}(\beta_1,\ldots,\beta_m)U_2.
\end{equation}

\end{proof}

\noindent It is also known that

\begin{coro}
The space of separable states $S(\mathcal{H})\subseteq
A(\mathcal{H})$ is compact.
\end{coro}

\begin{proof}
We remind the reader that Carathéodory's theorem states that if a
point $x$ of $\mathcal{R}^d$ lies in the convex hull of a set $P$,
there is a subset $P'$ of $P$ consisting of $d+1$ (or fewer)
points such that $x$ lies in the convex hull of $P'$.
Equivalently, $x$ lies in an $r-$simplex with vertices in $P$,
where $r \le d$. The result is named for Constantin Carathéodory,
who proved the theorem in 1911 for the case when $P$ is compact.
In 1914, Ernst Steinitz expanded Carathéodory's theorem for any
sets $P$ in $\mathcal{R}^d$. Let
$N=\dim_{\mathcal{R}}A(\mathcal{H})$. Then, it is seen by
Carathéodory's theorem \cite[17.1]{rocka} that the following
continuous map is surjective,

\begin{equation}
\triangle^{N+1}\times(\mathcal{S}_{0}(\mathcal{H}))^{N+1}\stackrel{\Gamma}{\longrightarrow}S(\mathcal{H})
\end{equation}

\begin{equation}
\Gamma((\lambda_0,\ldots,\lambda_N),(a_0,\ldots,a_{N}))=\sum_{i=0}^N
\lambda_i a_i.
\end{equation}

\noindent Hence $S(\mathcal{H})$ is a compact convex space.

\end{proof}

\noindent Entanglement witnesses are useful to characterize
entanglement \cite{bengtssonyczkowski2006,Horodeki-2009} and have
became a fundamental tool. They originate in geometry, being a
consequence of the Hahn-Banach theorem
\cite{bengtssonyczkowski2006}.
 This theorem is a central tool in functional analysis.
 It allows the extension of bounded linear functionals defined on a subspace of some vector space to the whole space,
 and it also guarantees that there are ``enough" continuous linear functionals defined on every normed vector
 space to make the study of the dual space ``interesting." Another version of the HB theorem,  known also
 either i) as the HB-Banach separation theorem or ii) the separating hyperplane theorem, has numerous uses in convex geometry, and interest us here.
 If $\rho\notin
\mathcal{S}(\mathcal{H})$, and because $\mathcal{S}(\mathcal{H})$
is closed and convex, then there exists a plane (and so a
functional) which separates $\rho$ and $\mathcal{S}(\mathcal{H})$.
This means that $\rho$ stands in one side of the plane and
$\mathcal{S}(\mathcal{H})$ on the other. The existence of such a
plane is closely linked to the existence of an observable $W$ (and
this is the entanglement witness) which satisfies that if
$\rho\in\mathcal{E}(\mathcal{H})$, then $\mbox{tr}(\rho W)<0$ and
$\mbox{tr}(\sigma)\geq 0$ for any
$\sigma\in\mathcal{S}(\mathcal{H})$. It is possible to show that
for any entangled state there exists an entanglement witness.
Also:

\begin{prop}
Let $\ell_\rho=\langle\rho,-\rangle$ an $\mathcal{R}$-linear
functional with $\rho\in\mathcal{A}(\mathcal{H})$. Then
$$\min_{S(\mathcal{H})}\ell_\rho=\min_{\mathcal{S}_{0}(\mathcal{H})}\ell_\rho=
\min\{\langle\rho,vv^\dag\otimes
ww^\dag\rangle\,|\,\|v\|=\|w\|=1\}=:m_\rho,$$
$$\max_{S(\mathcal{H})}\ell_\rho=\max_{\mathcal{S}_{0}(\mathcal{H})}\ell_\rho=
\max\{\langle\rho,vv^\dag\otimes
ww^\dag\rangle\,|\,\|v\|=\|w\|=1\}=:M_\rho.$$
\end{prop}

\begin{proof}
Suppose that $\lambda_0a_0+\ldots+\lambda_s a_s\in S(\mathcal{H})$
is the maximum of $\ell_\rho$ and consider the following linear
functional
$$(\alpha_0,\ldots,\alpha_s)\longrightarrow\sum_{i=0}^s \alpha_i \ell_\rho(a_i),\quad(\alpha_0,\ldots,\alpha_s)\in\triangle^{s+1}.$$
We know that its maximum is at $(\lambda_0,\ldots,\lambda_s)$ but
let us appeal to some ideas from the theory of faces discussed in
 \cite[p.162,163]{rocka}. We learn there that a linear functional over a convex set
$\triangle^{s+1}$, by definition, achieves its maximum over a face
($F=\triangle^{s+1}\cap H$) and given that every face contains
$0$-dimensional faces we have that every linear functional achieves
its maximum at some  $0$-dimensional face. Recall that ``a face of a
face" is just a face of the convex set $\triangle^{s+1}$ and that
the $0$-dimensional faces of $\triangle^{s+1}$ are its vertexes. In
conclusion, the linear functional considered above achieves its
maximum at a vertex. Then such maximum refers to a product state
(same for the minimum). For the second relationship invoked above we
need still to show that every separable state is a convex
combination of products of pure states, but this is a well
established fact (given $\rho$ separable, just express in diagonal
forma the components of products which appear on its decomposition).
\end{proof}

\noindent What is interesting regarding the above proposition is
the fact that we can effectively compute the numbers $m_\rho$ and
$M_\rho$, as demonstrated in  \cite{max}. Also, we can show that

\begin{rema}
Over the convex set $\mathcal{C}$ it is easy to prove that the
maximum and minimum of the lineal functional $\ell_\rho$ are the
biggest and smallest eigenvalues of $\rho$. Given that
$S(\mathcal{H})\subseteq\mathcal{C}(\mathcal{H})$ we have that
$\min_\mathcal{S}\ell_\rho\geq\lambda_{\min}$ and
$\max_\mathcal{S}\ell_\rho\leq\lambda_{\max}$. In particular, for
pure states we have $\max_\mathcal{S}\ell_{xx^\dag}\leq1$ and
$\min_\mathcal{S}\ell_{xx^\dag}\geq0$.
\end{rema}

\noindent Let us define

\begin{definition}
Let $e_i=(0,\ldots,0,1,0,\ldots,0)$ be a canonical column vector
and let $c_{kl}$ be the matrix with $1$ in place $kl$, i.e.
$c_{kl}=e_k e_l^\dag$. Let $x=\sum\lambda_i e_i$ be a vector such
that $\|x\|=1$ then
$$x^\dag\rho x=\sum_{kl} \overline{\lambda_k}\lambda_l e_k^\dag\rho e_l=
\sum_{kl} \overline{\lambda_k}\lambda_l \text{tr}(\rho c_{kl})=
\text{tr}(\rho \sum_{kl} \overline{\lambda_k}\lambda_l c_{kl})=
\text{tr}(\rho xx^\dag).$$ Note that $xx^\dag$ is a pure state
(positive and rank one matrix) and $\text{tr}(xx^\dag)=\|
x\|^2=1$. Let us define the sphere in $\mathcal{H}$
$$S_{1}(\mathcal{H})=\{x\in\mathcal{H}\,|\,\|x\|=1\}.$$
\end{definition}

\begin{prop}
\nd The space of states $\mathcal{C}$ is compact.
\end{prop}

\begin{proof}
\nd Let $n=\dim_\mathcal{R} \mathcal{H}$ and
$N=\dim_\mathcal{R}\mathcal{A}(\mathcal{H})$. We can give here two
proofs,
$$\triangle^{N+1}\times S_{1}(\mathcal{H})^{N+1}\stackrel{\Gamma_1}{\longrightarrow}\mathcal{C}(\mathcal{H}),\quad
\triangle^{N+1}\times\mathcal{U}(\mathcal{H})^n\times\triangle^n\stackrel{\Gamma_2}{\longrightarrow}\mathcal{C}(\mathcal{H}).$$
$$\Gamma_1(a_0,\ldots,a_N,x_0,\ldots,x_N)=\sum_{i=0}^N a_ix_ix_i^\dag.$$
$$\Gamma_2(a_0,\ldots,a_N,U_1,\ldots,U_n,d_1,\ldots,d_n)=\sum_{i=0}^N a_i U_i^\dag\text{diag}(d_1,\ldots,d_n)U_i.$$
Both are continuous surjective maps.
\end{proof}

\noindent Now we are in condition to characterize
$\mathcal{S}(\mathcal{H})$ in terms of special convex sets

\begin{theo}
$$\mathcal{C}=\bigcap_{x\in S_{1}(\mathcal{H})}\{\ell_{xx^\dag}\geq 0\},\quad
S(\mathcal{H})=\bigcap_{\rho\in\mathcal{A}(\mathcal{H})}\{m_\rho\leq\ell_\rho\leq
M_\rho\}.$$ In particular, if $s\in S(\mathcal{H})$ then
$\|s\|^2\leq M_s$.
\end{theo}

\begin{proof}
We only have to prove the second equality. If
$x\in\mathcal{A}(\mathcal{H})$ and $m_\rho\leq\ell_\rho(x)\leq
M_\rho$ for all $\rho\in\mathcal{A}(\mathcal{H})$, then $x\in
S(\mathcal{H})$. Suppose that $x\not\in S(\mathcal{H})$. Then, there
exists $\ell$ separating $x$ and $S(\mathcal{H})$. Assume that there
exist also $d\in\mathcal{R}$ such that

$$\ell(x)>d,\,\ell(s)\leq d\quad\forall s\in\mathcal{S}(\mathcal{H}).$$

Let $M$ be the maximum of $\ell$ over $\mathcal{S}(\mathcal{H})$.
Then $M\leq d$. Note that $\ell$ is one of the $\{\ell_\rho\}$ used
in the intersection operation above, so that we would have
$\ell(x)\leq M\leq d$, which is a contradiction! Accordingly,
$$\ell(x)<d,\,\ell(s)\geq d\quad\forall s\in\mathcal{S}(\mathcal{H}).$$
Now, let $m$ be the minimum of $\ell$ over $\mathcal{S}$. Then one
would have  $d\leq m$, which is again a contradiction!
\end{proof}

\nd Remark that any set of the form $\{m_\rho\leq\ell_\rho\leq
M_\rho\}\cap\mathcal{C}$ is always convex (and thus an element of
$\mathcal{L}_{\mathcal{C}}$). Consequently, the equality stated in
the above theorem

\begin{equation}\label{e:Scaracterizado1}
S(\mathcal{H})=\bigcap_{\rho\in\mathcal{A}(\mathcal{H})}\{m_\rho\leq\ell_\rho\leq
M_\rho\}
\end{equation}
\noindent may be recast in lattice theoretical terms as

\begin{equation}\label{e:Scaracterizado2}
S(\mathcal{H})=\bigwedge_{\rho\in\mathcal{A}(\mathcal{H})}\{m_\rho\leq\ell_\rho\leq
M_\rho\}\cap\mathcal{C}
\end{equation}

\noindent The difference between equations
(\ref{e:Scaracterizado1}) and (\ref{e:Scaracterizado2}) is both
subtle and important: one of them is expressed in lattice form. In
the previous theorem we aimed to characterize the separable states
via a consideration of all the linear functionals. In practice we
just want to know if a {\it given} state is separable. For such
query we will use a theorem on projections over convex sets
\cite[V.2]{brezis}.

\begin{prop}
\nd Let $\rho\in\mathcal{C}$ Then, there exists a linear functional
$\ell$ and a real number $M$ such that
$$\rho \not\in S(\mathcal{H})\iff \ell(\rho)\ge M.$$
\end{prop}

\begin{proof}
\nd If $\rho\not\in S(\mathcal{H})$ then the distance between $\rho$
and $S(\mathcal{H})$ is positive. Let $s\in S(\mathcal{H})$ be its
``projection" in the sense that the overlaps
$$\langle \rho-s,c\rangle\leq\langle \rho-s,s\rangle,\quad\forall c\in
S(\mathcal{H}).$$ Let $\ell:=\langle \rho-s,\cdot\rangle$ and
$M:=\langle \rho-s,s\rangle=\ell(s)$,
$$\ell(\rho)=\ell(\rho)-M+M=\langle \rho-s,\rho-s\rangle+M=\|
\rho-s\|^2+M>M.$$ Then $\ell$ separates $\rho$ and $S(\mathcal{H})$.
If $\rho\in S(\mathcal{H})$ then $\ell(\rho)=\ell(s)=M$.
\end{proof}

\noindent The following definition and subsequent proposition will
be useful for our separability criteria

\begin{definition}
If we identify $x\in S_{1}(\mathcal{H})$ with a matrix
$X\in\mathcal{C}^{n\times m}$, $X_{ij}=x_{m(i-1)+j}$, we have
$$x^\dag (v\otimes w)=v^\dag X w.$$
Let $\{v_i\}\subseteq\mathcal{H}_1$,
$\{w_i\}\subseteq\mathcal{H}_2$ and
$\sigma_1\geq\ldots\geq\sigma_r>0$ be the singular value
decomposition of $X$. It is known that the maximum of the bilinear
form $v^\dag Xw$ over $S_{1}(\mathcal{H}_1)\times
S_{1}(\mathcal{H}_2)$ is $\sigma_1$.
\end{definition}

\begin{prop}\label{p:separablecholevski}
$$S(\mathcal{H})\subseteq
\bigcap_{x\in
S_{1}(\mathcal{H})}\{\rho\in\mathcal{C}(\mathcal{H})\,|\,\langle
xx^\dag,\rho\rangle\leq\sigma_1(X)^2\}= \bigcap_{x\in
S_{1}(\mathcal{H})}\{0\leq\ell_{xx^\dag}\leq\sigma_1(X)^2\}.$$
\end{prop}

\begin{proof}
A general separable state is a convex combination of product
states $vv^\dag\otimes ww^\dag$,
$$\langle xx^\dag,vv^\dag\otimes ww^\dag\rangle=|\langle x,v\otimes w\rangle|^2\leq\sigma_1^2,\quad
\langle xx^\dag,v_1v_1^\dag\otimes
w_1w_1^\dag\rangle=\sigma_1^2\Longrightarrow
\max_{S(\mathcal{H})}\ell_{xx^\dag}=\sigma_1^2.$$ Note that
$\sigma_1^2=1$ if and only if $xx^\dag=v_1v_1^\dag\otimes
w_1w_1^\dag\in S(\mathcal{H})$. Let
$\rho\in\mathcal{C}(\mathcal{H})$ and assume that there exist a
state $xx^\dag$ such that $\langle
xx^\dag,\rho\rangle>\sigma_1(x)^2$. Then,  $\rho\not\in
S(\mathcal{H})$.
\end{proof}

\noindent The first inclusion of Proposition
\ref{p:separablecholevski} can be written in the language of
$\mathcal{L}_{C}$. This is so because, for a fixed $x\in S_{1}$,
the set

\begin{equation}
C_{x}=\{\rho\in\mathcal{C}(\mathcal{H})\,|\,\langle
xx^\dag,\rho\rangle\leq\sigma_1(X)^2\},
\end{equation}

\noindent is convex. In order to better appreciate this fact,
suppose that $\rho_{1}\in C_{x}$ and $\rho_{2}\in C_{x}$. Then, we
have $\langle xx^\dag,\rho_{1}\rangle\leq\sigma_1(x)^2$ and
$\langle xx^\dag,\rho_{2}\rangle\leq\sigma_1(x)^2$. Multiplying
the first inequality by $\lambda\in(0,1)$ and the second one by
$(1-\lambda)$, we easily find that $\langle
xx^\dag,(\lambda\rho_{1}+(1-\lambda)\rho_{2})\rangle\leq\sigma_1(x)^2$.
This proves that $C_{x}$ is convex.  Thus, for each $x\in S_{1}$,
$C_{x}\in \mathcal{L}_{\mathcal{C}}$. Accordingly, we can state
(using a lattice theoretical language) that

\begin{prop}\label{p:inequality}
$\mathcal{S}(\mathcal{H})\leq \bigwedge_{x\in
S_{1}(\mathcal{H})}C_{x}.$
\end{prop}

\noindent The above proposition shows that the set of separable
states is included in the conjunction of a collection of special
elements of $\mathcal{L}(\mathcal{C})$. This leads to  a new
(partial) separability criteria, because from Proposition
\ref{p:inequality} it rapidly follows that\\

\fbox{\parbox{6.0in}{\nd Given the state $\rho$, if there exists
$x\in S_{1}(\mathcal{H})$ such that $\rho\notin C_{x}$, then
$\rho\in\mathcal{E}(\mathcal{H})$.}} \vskip 3mm \nd

\noindent The above discussion can be easily rephrased to prove
the following theorem

\begin{theorem}
Use the Cholesky and the singular value decompositions to write
$\rho\in\mathcal{C}(\mathcal{H})$ as a sum
$\rho=\sum_{i=1}^s\lambda_i x_ix_i^\dag$, with $x_i^\dag x_j=0$,
$$\rho=LL^\dag=(U\Sigma V^\dag)(U\Sigma V^\dag)^\dag=U\Sigma^2 U^\dag\Longrightarrow$$
$$\langle x_jx_j^\dag,\rho\rangle=\langle x_jx_j^\dag,\sum_{i=1}^s\lambda_ix_ix_i^\dag\rangle=
\sum_{i=1}^s\lambda_i\langle
x_jx_j^\dag,x_ix_i^\dag\rangle=\lambda_j.$$ If for some $j$,
$\lambda_j>\sigma_1(x_j)^2$ then, $\rho\not\in S(\mathcal{H})$.
\end{theorem}

\section{The characterization of entanglement using informational invariants}\label{s:The Relationship for convex}

In this section we study the relationship between the lattice
$\mathcal{L}_{\mathcal{C}}$ of a system $S$ composed of subsystems
$S_{1}$ and $S_{2}$, and the lattices of its subsystems,
$\mathcal{L}_{\mathcal{C}1}$ and $\mathcal{L}_{\mathcal{C}2}$
respectively. As in \cite{extendedql}, we do this by concocting a
physical interpretation of the maps which can be defined between
them. Recall that we are working with spaces of arbitrary
dimension.

\subsection{Separable States (Going Up)}\label{s:going up}

Let us define:

\begin{definition}
Given $C_{1}\subseteq\mathcal{C}_{1}$ and
$C_{2}\subseteq\mathcal{C}_{2}$

\begin{equation}
C_1\otimes C_2:=\{\rho_{1}\otimes\rho_{2}\,|\,\rho_{1}\in
C_1,\rho_{2}\in C_2\}
\end{equation}

\end{definition}

\noindent Then, we define the map:

\begin{definition}
$$\Lambda:\mathcal{L}_{\mathcal{C}1}\times\mathcal{L}_{\mathcal{C}2}\longrightarrow\mathcal{L}_{\mathcal{C}}$$
$$(C_{1},C_{2})\longrightarrow \overline{conv(C_1\otimes C_2)}$$
where the bar denotes closure respect to norm.
\end{definition}

\noindent In the rest of this work we will implicitly use the
following proposition (see for example \cite{Convexsets}):

\begin{prop}
Let $S$ be a subset of a linear space $\mathcal{L}$. Then $x\in
conv(S)$ iff $x$ is contained in a finite dimensional simplex
$\Delta$ whose vertices belong to $S$.
\end{prop}

\noindent From equation (\ref{e:separablestates}) and definition \ref{d:convex
tensor product} it should be clear that

\begin{equation}
\Lambda(\mathcal{C}_{1},\mathcal{C}_{2})=\mathcal{S}(\mathcal{H})
\end{equation}

\noindent Definition \ref{d:convex tensor product} also implies that for all
$C_{1}\subseteq\mathcal{C}_{1}$ and
$C_{2}\subseteq\mathcal{C}_{2}$:

\begin{equation}
\Lambda(C_{1},C_{2})=\overline{C_{1}\widetilde{\otimes}\,C_{2}}
\end{equation}

\begin{prop}
Let $\rho=\rho_{1}\otimes\rho_{2}$, with
$\rho_{1}\in\mathcal{C}_{1}$ and $\rho_{2}\in\mathcal{C}_{2}$.
Then $\{\rho\}=\Lambda(\{\rho_{1}\},\{\rho_{2}\})$ with
$\{\rho_1\}\in\mathcal{L}_{C1}$, $\{\rho_2\}\in\mathcal{L}_{C2}$
and $\{\rho\}\in\mathcal{C}$.
\end{prop}

\begin{proof}
We already know that atoms are special elements of  lattices.
Thus,

\begin{equation}
\Lambda(\{\rho_{1}\},\{\rho_{2}\})
=conv(\{\rho_{1}\otimes\rho_{2}\})=\{\rho_{1}\otimes\rho_{2}\}=\{\rho\}
\end{equation}

\end{proof}

\begin{prop}
Let $\rho\in\mathcal{S(\mathcal{H})}$, the set of separable
states. Then, there exists $C\in\mathcal{L}_{\mathcal{C}}$,
$C_{1}\in\mathcal{L}_{\mathcal{C}_{1}}$ and
$C_{2}\in\mathcal{L}_{C_{2}}$ such that
$\rho\in C=\Lambda(C_{1},C_{2})$.
\end{prop}

\begin{proof}
Let $\{\rho_n\}_{n=1}^\infty\subseteq\mathcal{S(\mathcal{H})}$ be a sequence in the interior of $\mathcal{S(\mathcal{H})}$
such that $\rho_n\rightarrow\rho$, then
$\rho_n=\sum_{i}\lambda_{i}\phi^n_{i}\otimes\psi^n_{i}$, with
$\sum_{i}\lambda_{i}=1$ and $\lambda_{i}\geq 0$. Consider the
convex sets:

$$C_{1}=conv(\{\phi_{i}^{n}\}_{i,n})\in\mathcal{L_C}_{1},\quad
C_{2}=conv(\{\psi_{i}^{n}\}_{i,n})\in\mathcal{L_C}_{2},\quad
C=\Lambda(C_{1},C_{2})\in\mathcal{L_C}.$$

\noindent Clearly, $\phi_{i}^{n}\otimes\psi_{i}^{n}\in
C_{1}\otimes C_{2}$, and then $\rho_n\in C$ for all
$n\in\mathcal{N}$. Given that $C$ is closed, we have $\rho\in C$.
\end{proof}

\subsection{Projections Onto $\mathcal{L}_{\mathcal{C}_{1}}$ and $\mathcal{L}_{\mathcal{C}_{2}}$ (Going
Down)}\label{s:projectionsc}

Let us now study the projections onto
$\mathcal{L}_{\mathcal{C}_{1}}$ and $\mathcal{L}_{\mathcal{C}_{2}}$.
In the next proposition we will see that they are well defined.
Using the partial trace maps we can construct the induced
projections:

\begin{subequations}
\begin{equation}
\tau_{i}:\mathcal{L}_{\mathcal{C}}\longrightarrow
\mathcal{L}_{\mathcal{C}_{i}}
\end{equation}

\begin{equation}
C\mapsto \mbox{tr}_{i}( C )
\end{equation}
\end{subequations}

\noindent Then we can define the product map

\begin{subequations}
\begin{equation}
\tau:\mathcal{L}_{\mathcal{C}}\longrightarrow\mathcal{L}_{\mathcal{C}_{1}}\times\mathcal{L}_{\mathcal{C}_{2}}
\end{equation}

\begin{equation}
C\mapsto(\tau_{1}(C),\tau_{2}(C))
\end{equation}
\end{subequations}

\nd We use the same notation for $\tau$ and $\tau_{i}$ (though
they have different domains) as in \cite{extendedql} and section
\ref{s:New language}, and this should not introduce any
difficulty. We can prove the following about the image of
$\tau_{i}$.

\begin{prop}\label{lastausonsurjective}
The maps $\tau_{i}$ preserve the convex structure, i.e., they map
convex sets into convex sets.
\end{prop}

\begin{proof}
Let $C\subseteq \mathcal{C}$ be a convex set. Let $C_{1}$ be the
image of $C$ under $\tau_{1}$ (a similar argument holds for
$\tau_{2}$). Let us show that $C_{1}$ is convex. Let $\rho_{1}$
and $\rho'_{1}$ be elements of $C_{1}$. Consider
$\sigma_{1}=\alpha\rho_{1}+(1-\alpha)\rho'_{1}$, with
$0\leq\alpha\leq 1$. Then, there exists $\rho,\rho'\in\mathcal{C}$
such that:

\begin{equation}
\sigma_{1}=\alpha\mbox{tr}_1(\rho)+(1-\alpha)\mbox{tr}_1(\rho')=\mbox{tr}_{1}(\alpha\rho+(1-\alpha)\rho')
\end{equation}

\noindent where we have used the linearity of trace. Because of
convexity of $C$, $\sigma:=\alpha\rho+(1-\alpha)\rho'\in C$, and so,
$\sigma_{1}=\mbox{tr}_{1}(\sigma)\in C_{1}$.
\end{proof}

\begin{prop}\label{lastausonsurjective1}
The functions $\tau_{i}$ are surjective and preserve the
$\vee$-operation. They are not injective.
\end{prop}

\begin{proof}
Take the convex set $C_{1}\in\mathcal{L}_{\mathcal{C}_{1}}$.
Choose an arbitrary element of $\mathcal{C}_{2}$, say $\rho_{2}$.
Now consider the following element of $\mathcal{L}_{\mathcal{C}}$

\begin{equation}
C=C_{1}\otimes\rho_{2}
\end{equation}

\noindent $C$ is convex, and so belongs to
$\mathcal{L}_{\mathcal{C}}$, because if
$\rho\otimes\rho_{2},\sigma\otimes\rho_{2}\in C$, then any convex
combination
$\alpha\rho\otimes\rho_{2}+(1-\alpha)\sigma\otimes\rho_{2}=(\alpha\rho+(1-\alpha)\sigma)\otimes\rho_{2}\in
C$ (where we have used the convexity of $C_{1}$). It is clear that
$\tau_{1}(C)=C_{1}$, because if $\rho_{1}\in C_{1}$, then
$\mbox{tr}_1(\rho_{1}\otimes\rho_{2})=\rho_{1}$. So, $\tau_{1}$ is
surjective. On the other hand, the arbitrariness of $\rho_{2}$
implies that it is not injective. An analogous argument follows for $\tau_2$.\\
Let us see that $\tau_i$ preserves the $\vee$-operation. Let $C$
and $C'$ be convex subsets of $\mathcal{C}$. We must compute
$\mbox{tr}_{2}(C\vee C'))=\mbox{tr}_{2}(conv(C,C'))$. We ought to
show that this is the same as
$conv(\mbox{tr}_{2}(C),\mbox{tr}_{2}(C'))$. Take $x\in
conv(\mbox{tr}_{2}(C),\mbox{tr}_{2}(C'))$. Then
$x=\alpha\mbox{tr}_{2}(\rho)+(1-\alpha)\mbox{tr}_{2}(\rho')$, with
$\rho\in C$, $\rho'\in C'$ and $0\leq\alpha\leq 1$. Using the
linearity of trace, $x=\mbox{tr}_{2}(\alpha\rho+(1-\alpha)\rho')$.
$\alpha\rho+(1-\alpha)\rho'\in conv(C,C')$, and so,
$x\in\mbox{tr}_{2}(conv(C,C'))$. Hence we have

\begin{equation}
conv(\mbox{tr}_{2}(C),\mbox{tr}_{2}(C'))\subseteq\mbox{tr}_{2}(conv(C,C'))
\end{equation}

\noindent In order to prove the other inclusion, take
$x\in\mbox{tr}_{2}(conv(C,C'))$. Then,

\begin{equation}
x=\mbox{tr}_{2}(\alpha\rho+(1-\alpha)\rho')=\alpha\mbox{tr}_{2}(\rho)+(1-\alpha)\mbox{tr}_{2}(\rho')
\end{equation}

\noindent with $\rho\in C_{1}$ and $\rho'\in C'$. Note that
$\mbox{tr}_{2}(\rho)\in\mbox{tr}_{2}(C)$ and
$\mbox{tr}_{2}(\rho')\in\mbox{tr}_{2}(C')$. This proves that

$$\mbox{tr}_{2}(conv(C,C'))\subseteq conv(\mbox{tr}_{2}(C),\mbox{tr}_{2}(C'))$$

\end{proof}

\noindent Let us now consider the $\wedge$-operation. If
$x\in\tau_i(C\wedge C')=\tau_i(C\cap C')$ then $x=\tau_i(\rho)$
with $\rho\in C\cap C'$. But, if $\rho\in C$, then
$x=\tau_i(\rho)\in\mbox{tr}_i(C)$. As $\rho\in C'$ as well, a
similar argument shows that $x=\tau_i(\rho)\in\mbox{tr}_i(C')$.
Then,  $x\in\tau_i(C)\cap \tau_i(C')$ and

\begin{equation}
\tau_i(C\cap C')\subseteq\tau_i(C)\cap \tau_i(C'),
\end{equation}
\noindent which is tantamount to

\begin{equation}
\tau_i(C\wedge C')\leq\tau_i(C)\wedge\tau_i(C').
\end{equation}

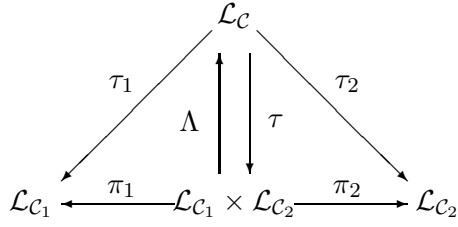
\begin{figure}\label{f:maps2}
\begin{center}
\unitlength=1mm
\begin{picture}(5,5)(0,0)
\put(-3,23){\vector(-1,-1){20}} \put(3,23){\vector(1,-1){20}}
\put(-2,4){\vector(0,2){16}} \put(2,20){\vector(0,-2){16}}
\put(8,0){\vector(3,0){15}} \put(-8,0){\vector(-3,0){15}}
\put(0,25){\makebox(0,0){$\mathcal{L}_{\mathcal{C}}$}}
\put(-27,0){\makebox(0,0){${\mathcal{L}_{\mathcal{C}_{1}}}$}}
\put(27,0){\makebox(0,0){${\mathcal{L}_{\mathcal{C}_{2}}}$}}
\put(0,0){\makebox(0,0){${\mathcal{L}_{\mathcal{C}_{1}}\times\mathcal{L}_{\mathcal{C}_{2}}}$}}
\put(-1,11){\makebox(-10,0){$\Lambda$}}
\put(-1,11){\makebox(13,0){$\tau$}}
\put(-15,16){\makebox(0,0){$\tau_{1}$}}
\put(15,16){\makebox(0,0){$\tau_{2}$}}
\put(-15,2){\makebox(0,0){$\pi_{1}$}}
\put(15,2){\makebox(0,0){$\pi_{2}$}}
\end{picture}
\caption{The different maps between
$\mathcal{L}_{\mathcal{C}_{1}}$, $\mathcal{L}_{\mathcal{C}_{2}}$,
${\mathcal{L}_{\mathcal{C}_{1}}\times\mathcal{L}_{\mathcal{C}_{2}}}$,
and $\mathcal{L}_{\mathcal{C}}$}
\end{center}
\end{figure}

\noindent These sets are not,  in general, equal. The following
example illustrates the assertion. Take
$\{\rho_{1}\otimes\rho_{2}\}\in\mathcal{L}$ and
$\{\rho_{1}\otimes\rho'_{2}\}\in\mathcal{L}$, with
$\rho'\neq\rho$. It is clear that
$\{\rho_{1}\otimes\rho_{2}\}\wedge\{\rho_{1}\otimes\rho'_{2}\}=\mathbf{0}$
and so,
$\tau_{1}(\{\rho_{1}\otimes\rho_{2}\}\wedge\{\rho_{1}\otimes\rho'_{2}\})=\mathbf{0}$.
On the other hand,
$\tau_{1}(\{\rho_{1}\otimes\rho_{2}\})=\{\rho_{1}\}=\tau_{1}(\{\rho_{1}\otimes\rho'_{2}\})$,
and  then
$\tau_{1}(\{\rho_{1}\otimes\rho_{2}\})\wedge\tau_{1}(\{\rho_{1}\otimes\rho'_{2}\})=\{\rho_{1}\}$.
A similar reasoning holds for the $\neg$-operation.

\subsection{Geometrical Characterization of
Entanglement}\label{s:entanglement}

\nd We have shown that it is possible to extend
$\mathcal{L}_{v\mathcal{N}}$ in order to deal with statistical
mixtures and that $\mathcal{L}$ and $\mathcal{L}_{\mathcal{C}}$
are possible extensions. It would be interesting to search for a
characterization of entanglement within this framework. Let us see
first what happens with the functions $\Lambda\circ\tau$ and
$\tau\circ\Lambda$. We have:

\begin{prop}\label{subir bajar}
$\tau\circ\Lambda(C_1,C_2)=(C_1,C_2)$ for every closed convex sets
$C_1\subseteq\mathcal{C}_1$ and $C_2\subseteq\mathcal{C}_2$.
\end{prop}

\begin{proof}
$$\tau_1(\Lambda(C_{1},C_{2}))=\tau_1(\overline{conv(C_{1}\otimes C_{2}}))=
\mbox{tr}_1(\overline{conv(C_{1}\otimes C_{2}}))=\overline{C_{1}}=C_1$$
$$\tau_2(\Lambda(C_{1},C_{2}))=\tau_2(\overline{conv(C_{1}\otimes C_{2}}))=
\mbox{tr}_2(\overline{conv(C_{1}\otimes
C_{2}}))=\overline{C_{2}}=C_2$$ Then,
$\tau(\Lambda(C_{1},C_{2}))=(C_{1},C_{2})$.
\end{proof}

\nd Again, as in \cite{extendedql}, if we take into account simple
physical considerations, $\Lambda\circ\tau$ is not the identity
function,  because when we take partial traces we face the risk of
losing information, that will not be recovered when we multiply
states. Thus we reach the same conclusion  as before
\cite{extendedql}: \emph{``going down and then going up is not the
same as going up and then going down''}. We depict w the pertinent
maps in Figure \ref{f:maps2}. How is this stuff related to
entanglement? If we restrict $\Lambda\circ\tau$ to the set of
product states, then it does reduce itself  to the identity
function. Indeed, if $\rho=\rho_{1}\otimes\rho_{2}$, then:

\begin{equation}
\Lambda\circ\tau(\{\rho\})=\{\rho\}.
\end{equation}

\noindent On the other hand, it should be clear that if $\rho$ is
an entangled state

\begin{equation}\label{entangled equation}
\Lambda\circ\tau(\{\rho\})\neq\{\rho\},
\end{equation}
\noindent because
$\Lambda\circ\tau(\{\rho\})=\{\mbox{tr}_{2}(\rho)\otimes\mbox{tr}_{1}(\rho)\}\neq\{\rho\}$
for any entangled state. This property can be regarded as a signpost
for entanglement. There are mixed states which are not product
states. Thus, entangled states are not the only ones satisfying
equation (\ref{entangled equation}). What is the condition satisfied
for a general mixed state? The following proposition summarizes the
preceding considerations.

\begin{prop}\label{subir bajar1}
If $\rho$ is a separable state, then there exists a convex set,
$S_{\rho}\subseteq\mathcal{S}(\mathcal{H})$ such that $\rho\in
S_{\rho}$ and $\Lambda\circ\tau(S_{\rho})=S_{\rho}$. More
generally, for a convex set $C\subseteq \mathcal{S}(\mathcal{H})$,
there exists a convex set $S_{C}\subseteq\mathcal{S}(\mathcal{H})$
such that $\Lambda\circ\tau(S_{C})=S_{C}$. For a product state, we
can choose $S_{\rho}=\{\rho\}$. If $\rho$ can be written as a
finite convex sum of product states, then the convex set
$S_{\rho}$ can be taken as a polytope. On the other hand, for any
$C\in\mathcal{L}_{\mathcal{C}}$ which has at least {\it one}
non-separable state,  there is NO convex set $S$ such that
$C\subseteq S$ and $\Lambda\circ\tau(S)=S$.
\end{prop}

\begin{proof}
\noindent \textit{Product case.} We have already seen above that
if $\rho$ is a product state, then
$\Lambda\circ\tau(\{\rho\})=\{\rho\}$, and so $S_{\rho}=\{\rho\}$.

\noindent \textit{Finite combination case.} If $\rho$ can be
written as a finite convex combination of product states, then
there exists $\rho^{A}_{k}\in \mathcal{C}_{1}$,
$\rho^{B}_{k}\in\mathcal{C}_{2}$ and $\alpha^{i}_{k}\geq 0,
\sum_{k=1}^{N}\alpha^{i}_{k}=1$ such that

\begin{equation}\label{e:separablerho}
\rho=\sum_{k=1}^{N}\alpha_{k}\rho^{A}_{k}\otimes\rho^{B}_{k}
\end{equation}

Define first
\begin{equation}
S_\rho=conv(\{\rho^{A}_{k}\otimes\rho^{B}_{l}\})
\end{equation}

\noindent $S_\rho$ is the closed set of all convex combinations of
products of the elements  appearing in the decomposition of
$\rho$. It should be clear that $\rho\in S_\rho$. Let us compute
$\Lambda\circ\tau(S_\rho)$.
\begin{equation}\label{e:cuenta}
\mbox{tr}_{1}(\sigma)=\sum_{k=1}^{N}(\sum_{l=1}^{N}\lambda_{kl})
\rho^{A}_{k}=\sum_{k=1}^{N}\mu_{k} \rho^{A}_{k},\quad \mu_{k}:=\sum_{l=1}^{N}\lambda_{kl}.
\end{equation}

\noindent In an analogous way we may show that an element of
$\tau_{2}(S_\rho)$ has the form $\sum_{l=1}^{N}\nu_{l}
\rho^{A}_{l}$ with $\nu_{l}=\sum_{k=1}^{N}\lambda_{kl}$. Note that
$\sum_{k=1}^{N}\mu_{k}=\sum_{l=1}^{N}\nu_{l}=1$. In order to
compute $\Lambda(\tau_{1}(S_\rho),\tau_{2}(S_\rho))$ we must
construct the convex hull of the set
\begin{equation}
\tau_{1}(S_\rho)\otimes\tau_{2}(S_\rho)=\{\sigma_{1}\otimes\sigma_{2}|\sigma_{1}\in\tau_{1}(S_\rho),\sigma_{2}\in\tau_{2}(S_\rho)\}=
\{\sum_{k,l=1}^{N}\mu_{k}\nu_{l}\rho^{A}_{k}\otimes\rho^{B}_{l}\},
\end{equation}
\noindent and we conclude that
\begin{equation}\label{e:equality}
\Lambda\circ\tau(S_\rho)=conv(\{\sum_{kl=1}^{N}\mu_{k}\nu_{l}\rho^{A}_{k}\otimes\rho^{B}_{l}\})=conv(\{\rho^{A}_{k}\otimes\rho^{B}_{l}\})=S_\rho.
\end{equation}

\noindent It is apparent that $S_\rho$ is a polytope.

\vskip 3mm \noindent \textit{Limit point case.} There is still
another possibility. Namely if $\rho$ cannot be written as in
(\ref{e:separablerho}), but there exists
$\rho_{ik}^{A}\in\mathcal{C}_{1}$,
$\rho_{ik}^{B}\in\mathcal{C}_{2}$ and $\alpha_{ik}\geq 0,
\sum_{k=1}^{N_{i}}\alpha_{ik}=1$ such that
$\rho^{i}=\sum_{k=1}^{N_{i}}\alpha_{ik}\rho^{A}_{ik}\otimes\rho^{B}_{ik}$
converges to $\rho$ as $i$ goes to infinity. Consider the set of
all possible products of states which appear in the decomposition
of the $\rho^{i}$, namely

\begin{equation}
S_{0}:=\{\rho^{A}_{ik}\otimes\rho^{B}_{i'l}, \mbox{for
all}\,\,i,i',k,l\}
\end{equation}

\noindent and define the closure of its convex hull as

\begin{equation}\label{d:Minvariantset}
S_\rho:=\overline{conv(S_{0})}
\end{equation}

\noindent Remember that convex hull means only taking finite sums.
It is clear that $\rho\in S_\rho$ and that $S_\rho$ is convex by
construction (the closure of a convex set is also convex). Let us
see what happens when we apply $\Lambda\circ\tau$ to $S_\rho$,

\begin{equation}\label{e:lambdatauM}
\Lambda(\tau(S_\rho))=\overline{conv(\tau_{1}(S_\rho)\otimes\tau_{2}(S_\rho))}
\end{equation}

\noindent Any of the $\rho^{A}_{ik}$ belong to $\tau_{1}(S_\rho)$
(the same for $\rho^{B}_{i'k}$ and $\tau_{2}(S_\rho)$). Then, it
is clear that $S_{0}\subseteq
conv(\tau_{1}(S_\rho)\otimes\tau_{2}(S_\rho))$. As
$conv(\tau_{1}(S_\rho)\otimes\tau_{2}(S_\rho))$ is convex and the
closure of sets preserves the inclusion, then we have
$S_\rho\subseteq\Lambda\circ\tau(S_\rho))$ (look at Equation
(\ref{d:Minvariantset})). On the other hand, any element
$\rho_{1}$ of $\tau_{1}(S_\rho)$ can be written as a finite sum
$\rho_{1}=\sum\alpha_{ik}\rho^{A}_{ik}$ ($\sum\alpha_{ik}=1$,
$\alpha_{ik}\geq 0$) or as a limit of such finite sums (we are
using a property of partial traces $\mbox{tr}_{i}$: they
 are
continuous linear maps). The same happens for an element
$\rho_{2}\in\tau_{2}(S_\rho)$ (taking the tensor product of
density operators produces a continuous map). Then, any element of
$\tau_{1}(S_\rho)\otimes\tau_{2}(S_\rho)$ may be written as a
finite sum
$\sum\alpha_{ik}\beta_{i'l}\rho^{A}_{ik}\otimes\rho^{B}_{i'l}$ or
as a limit of such sums. This means that any element of
$\tau_{1}(S_\rho)\otimes\tau_{2}(S_\rho)$ is also an element of
$S_\rho$. As $S_\rho$ is convex and closed by construction, we
will have
$\Lambda\circ\tau(S_\rho)=\overline{conv(\tau_{1}(S_\rho)\otimes\tau_{2}(S_\rho))}\subseteq
S_\rho$, which proves that $\Lambda\circ\tau(S_\rho)=S_\rho$.\\
\vskip 3mm

\noindent The space of separable states $\mathcal{S}(\mathcal{H})$
is a convex set. Let us see that it is invariant under
$\Lambda\circ\tau$. First of all, we know that
$\mathcal{S}(\mathcal{H})$ is formed by the closure of all possible
convex combinations of products of the form
$\rho_{1}\otimes\rho_{2}$, with $\rho_{1}\in\mathcal{C}_{1}$ and
$\rho_{2}\in\mathcal{C}_{2}$. But each one of these tensor products,
$\Lambda\circ\tau(\{\rho_{1}\otimes\rho_{2}\})=\{\rho_{1}\otimes\rho_{2}\}$,
belongs to $\Lambda\circ\tau(\mathcal{S}(\mathcal{H}))$. Given that
$\Lambda\circ\tau(\mathcal{S}(\mathcal{H}))$ is a closed convex set,
we have
$\Lambda\circ\tau(\mathcal{S}(\mathcal{H}))\supseteq\mathcal{S}(\mathcal{H})$.
On the other hand we know that the image of $\Lambda\circ\tau$ is
always separable, so we can conclude that

\begin{equation}
\Lambda\circ\tau(\mathcal{S}(\mathcal{H}))=\mathcal{S}(\mathcal{H})
\end{equation}

\noindent Now, consider $C\in\mathcal{L}_{\mathcal{C}}$ such that
there exists $\rho\in C$, with $\rho$ nonseparable. Given that
$\Lambda\circ\tau(S)\subseteq \mathcal{S}(\mathcal{H})$ for all
$S\in\mathcal{L}_{\mathcal{C}}$, it could never happen that there
exists $S\in\mathcal{L}_{\mathcal{C}}$ such that $C\subseteq S$
and $\Lambda\circ\tau(S)=S$.

\end{proof}

\noindent From the last proposition, we conclude that there exists
an interesting  property which the convex subsets of separable
states satisfy, while convex subsets which include non-separable
states do not. This ``existence theorem" motivates the following
definition for the proposition $C\in\mathcal{L}_{\mathcal{C}}$:

\begin{definition}
$C\in\mathcal{L}_{\mathcal{C}}$ is a separable proposition if
there exists $S_{C}\in\mathcal{L}_{\mathcal{C}}$ such that
$\Lambda\circ\tau(S_{C})=S_{C}$ and $C\subseteq S_{C}$. Otherwise,
 it is a non-separable or entangled proposition. The definition
is equivalent to the statement
$C\subseteq\mathcal{S}(\mathcal{H})$.
\end{definition}

\noindent Another conclusion of proposition \ref{subir bajar1} is
that a density matrix $\rho$ is separable iff there exists a
convex set $S_{\rho}$ such that $\rho\in S_{\rho}$ and
$\Lambda\circ\tau(S_{\rho})=S_{\rho}$. Thus, proposition
\ref{subir bajar} also provides an entanglement criterium {\bf
which includes the
infinite dimensional case} (see also \cite{Holik-Plastino-2011a}):\\

\fbox{\parbox{6.0in}{
\begin{equation}\label{e:ourcriteria}
\rho\in\mathcal{S}(\mathcal{H}) \Longleftrightarrow\text{ there
exists a convex set }S_{\rho}\text{ with } \rho\in S_{\rho}\text{
such that }\Lambda\circ\tau(S_{\rho})=S_{\rho}.
\end{equation}
}}\vskip 3mm \nd

\subsection{An unifying generalization for the entanglement of mixed states}

\noindent In the last section we have introduced a new
separability criterium which is also valid for the infinite
dimensional case. Now we proceed to an important issue regarding
pure states (see also the discussions in
\cite{Holik-Plastino-2011a}). It is a well known fact that pure
states are separable, if and only if they are product states. This
means that $|\psi\rangle\langle\psi|$ will be separable if and
only if there exist $|\varphi_{1}\rangle$ and
$|\varphi_{2}\rangle$ such that
$|\psi\rangle=|\varphi_{1}\rangle\otimes|\varphi_{2}\rangle$. This
implies that the state $|\psi\rangle\langle\psi|$ is invariant
under the map

\begin{eqnarray}
\Omega:\mathcal{C}\longrightarrow\mathcal{C}\nonumber\\
\rho\mapsto\rho^{A}\otimes\rho^{B},
\end{eqnarray}
\noindent and this in turn means that

\begin{equation}\label{e:purecriteria}
|\psi\rangle\langle\psi|\in\mathcal{S}(\mathcal{H})\Longleftrightarrow\Omega(|\psi\rangle\langle\psi|)=|\psi\rangle\langle\psi|.
\end{equation}

\noindent Such simple separability criterium for the pure case is
unfortunately invalid for the mixed case. In what follows we show
that our separability criteria allows for an interesting unifying
generalization.

\vskip 3mm \nd First of all, notice that only product states
exhibit the property of being invariant under $\Omega$. Is there
any generalization of $\Omega$ and of the notion of product
states? Let us look in more detail to the  invariance under
$\Omega$-property. In mathematical terms, suppose that a state
$\rho$ satisfies

\begin{equation}
\Omega(\rho)=\rho.
\end{equation}

\noindent This is equivalent to stating that $\rho$ can be fully
recovered from its reduced states by using local operations. It is
easy to show that the function $\Lambda\circ\tau$ satisfies

\begin{equation}
\Lambda\circ\tau(\Lambda\circ\tau(C))=\Lambda\circ\tau(C),
\end{equation}
\noindent and this is equivalent to

\begin{equation}\label{e:lambdataucuadrado}
(\Lambda\circ\tau)^{2}=\Lambda\circ\tau,
\end{equation}

\noindent a property that $\Omega$ also satisfies, i.e.

\begin{equation}\label{e:omegacuadrado}
\Omega^{2}=\Omega,
\end{equation}

\noindent as may be easily checked out. It is trivially shown that,
when restricted to ``one point" convex subsets of the form
$\{\rho\}$ (an arbitrary state), $\Lambda\circ\tau$ coincides with
$\Omega$, that is

\begin{equation}\label{e:coincidence}
\Lambda\circ\tau(\{\rho\})=\{\rho^{A}\otimes\rho^{B}\}=\{\Omega(\rho)\}.
\end{equation}

\noindent Equations (\ref{e:lambdataucuadrado}),
(\ref{e:omegacuadrado}) and (\ref{e:coincidence}) clearly suggest
that $\Lambda\circ\tau$ is a suitable generalization of $\Omega$
to arbitrary convex subsets (a single state being a particular
case of one point convex sets). The separability criterium
presented in section \ref{s:The Relationship for convex} provides
the clue for generalizing product states to convex subsets, i.e.,
the convex set generalization of a product state $C$ will satisfy

\begin{equation}\label{e:CSSinvariance}
\Lambda\circ\tau(C)=C,
\end{equation}

\noindent and this reduces to the  the separability properties
defined in section \ref{s:The Relationship for convex}. The special
subsets of $\mathcal{C}$ that we are concerned with exhibit the
following property: they can be fully recovered via all possible
tensor products and mixtures of its sets of reduced states. More
specifically, given a convex set $C$ satisfying
(\ref{e:CSSinvariance}), it can be recovered  from the sets of its
reduced states, namely $\tau_{1}(C)$ and $\tau_{2}(C)$ via all
possible tensor products and all possible convex mixtures. In
physical  terms this means that they can be recovered using
classical and local operations (just adding systems via all possible
tensor products and then considering all possible mixtures of the
resulting states). The content of  this discussion is compactly
encapsulated into equation (\ref{e:CSSinvariance}). Convex subsets
with this property where termed \emph{Convex Invariant Subsets}
(CSS) in \cite{Holik-Plastino-2011a}. Now it should be clear that
CSS are proper generalizations of product states to arbitrary convex
subsets.

\noindent We can now generalize equation (\ref{e:purecriteria}) to
arbitrary states as follows. Separability criterium
\ref{e:ourcriteria} implies that a state $\rho$ is separable iff it
belongs to a CSS $C$ such that $\Lambda\circ\tau(C)=C$. The analogy
with the pure-states case is clear if we effect the identification
$\rho\longrightarrow\{\rho\}$ (i.e., the state considered as an
element to the state considered as a particular case of convex
subset).

\noindent We have thus shown that the map $\Lambda\circ\tau$ is a
suitable generalization of $\Omega$. The sets invariant under
$\Omega$ are product states and the sets invariant under
$\Lambda\circ\tau$ are CSS, a suitable generalization of product
states. We may now generalize equation (\ref{e:purecriteria}) to any
state as follows:

\begin{equation}\label{e:puregeneralizedcriteria}
\rho\in\mathcal{S}(\mathcal{H})\Longleftrightarrow \exists C\text{
(a CSS), }\Lambda\circ\tau(C)=C,
\end{equation}
\noindent a neat extension of  (\ref{e:purecriteria}). For the
finite dimensional case the analogy is stronger still: criterium
(\ref{e:purecriteria}) can be rephrased using von Neumann's entropy

\begin{equation}
S(\rho)=-\mbox{tr}(\rho\ln(\rho)),
\end{equation}

\noindent as follows:

\begin{equation}\label{e:purecriteriaentropy}
\rho\in\mathcal{S}(\mathcal{H})\Longleftrightarrow
S(\rho^{A})=0=S(\rho^{B}),
\end{equation}

\noindent where $\rho^{A}$ and $\rho^{B}$ are the reduced states of
$\rho$. As products of pure states generate (in the convex sense)
all separable states, it is possible to show that the CSS-criterium
may be, in particular, chosen to be generated by products of pure
states. von Neumann's entropy reaches its
minimum value in such an instance. Summing up:\\

\fbox{\parbox{6.0in}{\nd
$\rho\in\mathcal{S}(\mathcal{H})$$\Longleftrightarrow$ there exist
$C$ such that $\rho\in C$ and $\Lambda\circ\tau(C)=C$
$\Longleftrightarrow$ (finite dimension) there exist $C$ such that
$\rho\in C$, $\Lambda\circ\tau(C)=C$ and
$\inf\{S(\sigma)\,|\,\sigma\in C\}=0.$}} \vskip 3mm

\noindent The fact that the above structure may be found for
arbitrary states is a clear conceptual simplification for the
characterization of entanglement, providing a unifying framework
which generalizes (\ref{e:purecriteria}) to arbitrary states. In the
following section we outline how this geometrical structure extends
to arbitrary COMs, and thus to any statistical theory.

\section{Entanglement and separability in arbitrary convexity models}\label{s:PositiveMaps}

In Section \ref{s:probabilities} we reviewed how to construct a
general setting for convex operational models out of which the
quantum case was a particular example. In this section we study how
to extend our geometrical formulation of entanglement to arbitrary
statistical models. \vskip 3mm

\nd Given two convex operational models
$(\mathbf{A},\mathbf{A}^{\sharp},u_{\mathbf{A}})$ and
$(\mathbf{B},\mathbf{B}^{\sharp},u_{\mathbf{B}})$, a morphism
between them will be given by a positive linear map
$\phi:\mathbf{A}\rightarrow \mathbf{B}$ such that the linear
adjoint map $\phi^{\ast}:\mathbf{B}^{\ast}\rightarrow
\mathbf{A}^{\ast}$ is positive with respect to the cones
$\mathbf{A}_{+}^{\sharp}$ and $\mathbf{B}_{+}^{\sharp}$.

\nd A link between (or process from))   $\mathbf{A}$ -
$\mathbf{B}$ will be represented by a morphism
$\phi:\mathbf{A}\rightarrow \mathbf{B}$ such that, for every state
$\alpha\in\Omega_{\mathbf{A}}$, $u_{\mathbf{B}}(\phi(\alpha))\leq
1$ (this is a normalization condition).
$u_{\mathbf{B}}(\phi(\alpha))$ will represent the probability that
the process represented by $\phi$ take place. For the special case
of quantum mechanics, we will show that the above processes
preserve the convex structure of the cone of positive self adjoint
operators. Also, we demonstrate that when the processes preserve
trace (i.e., when they map density operators into density
operators and thus represent quantum evolutions), they will also
preserve the lattice structure of $\mathcal{L}_{\mathcal{C}}$.

\nd In the preceding Section we saw how to characterize
entanglement and separability using maps between elements of
$\mathcal{L}_{\mathcal{C}}$ and $\mathcal{L}_{\mathcal{C}_{i}}$.
The interesting point here is that the most salient feature of our
lattices is their convex structure, and this will allow us to
extend the notions of entanglement and separability to any COM.
This is done as follows. In \cite{Beltrametti.Varadarajan-2000}
extensions of COM's are studied (we review here their definition
of extension  slightly modifying the reference's notation). A COM
$(\mathbf{\mathbf{C}},\mathbf{C}^{\sharp},u_{\mathbf{C}})$ will be
said to be an extension of
$(\mathbf{A},\mathbf{A}^{\sharp},u_{\mathbf{A}})$ if there exists
a morphism $\phi:\mathbf{C}\rightarrow \mathbf{A}$ which is
surjective.

\nd In order to look for a generalization of entanglement which
captures the results of previous Sections we must look at triads
of COM's $(\mathbf{C},\mathbf{C}^{\sharp},u_{\mathbf{C}})$,
$(\mathbf{C}_{1},\mathbf{C}_{1}^{\sharp},u_{\mathbf{C}_{1}})$, and
$(\mathbf{C}_{2},\mathbf{C}^{\sharp}_{2},u_{\mathbf{C}_{2}})$,
such that there exist two morphisms $\phi_{1}$ and $\phi_{2}$ with
$(\mathbf{C},\mathbf{C}^{\sharp},u_{\mathbf{C}})$ an extension of
$(\mathbf{C_{1}},\mathbf{C_{1}}^{\sharp},u_{\mathbf{C_{1}}})$ and
$(\mathbf{C_{2}},\mathbf{C_{2}}^{\sharp},u_{\mathbf{C_{2}}})$. It
is clear that $\phi=(\phi_{1},\phi_{2})$ may be considered as the
best candidate for a generalization of $\tau$. Now, if we want an
analogue of $\Lambda$, we must demand additional requirements. We
are looking for a map $\Psi$ with  the following property. $\Psi$
maps any pair of non-empty convex subsets $(C_{1},C_{2})$ of
$\mathbf{C}_{1}\times \mathbf{C}_{2}$ into a non-empty convex
subset $C$ of $\mathbf{C}$ with this particular property: for any
$c\in C$, we must have $\phi_{1}(c)\in C_{1}$ and $\phi_{2}(c)\in
C_{2}$. Such property guarantees that for any pair of states
$a_{1}$ and $a_{2}$ there will always exist at least one state
$c\in \mathbf{C}$ such that $\phi_{1}(c)=a_{1}$ and
$\phi_{2}(c)=b_{1}$. Why? Because if $C_{1}=\{c_{1}\}$ and
$C_{2}=\{c_{2}\}$, then we must have $\phi_{1}(c)=a_{1}$ and
$\phi_{2}(c)=b_{2}$, which guarantees  that for any states $c_{1}$
and $c_{2}$ there will always exist a state $c$ for which $c_{1}$
and $c_{2}$, respectively, are the reduced states relative to the
maps $\phi_{i}$. As the maps $\phi_{i}$ are morphisms, using them
it is possible to define canonically induced functions on convex
subsets, and them to map convex subsets of $\mathbf{C}$ into
convex subsets of $\mathbf{C_{i}}$ (there is an analogy  with the
earlier language involving  $\tau_{i}$'s and partial traces). With
some abuse of notation we will  keep calling them $\phi_{i}'s$,
without undue harm.

Summing up:

\begin{definition}

A triad $(\mathbf{C},\mathbf{C}^{\sharp},u_{\mathbf{C}})$,
$(\mathbf{C}_{1},\mathbf{C}_{1}^{\sharp},u_{\mathbf{C}_{1}})$, and
$(\mathbf{C}_{2},\mathbf{C}^{\sharp}_{2},u_{\mathbf{C}_{2}})$ will
be called a \emph{triple compound system} if

\begin{enumerate}

\item There exist morphisms $\phi_{1}$ and $\phi_{2}$ such that
$(\mathbf{C},\mathbf{C}^{\sharp},u_{\mathbf{C}})$ is an extension
of $(\mathbf{C}_{1},\mathbf{C}_{1}^{\sharp},u_{\mathbf{C}_{1}})$
and $(\mathbf{C}_{2},\mathbf{C}^{\sharp}_{2},u_{\mathbf{C}_{2}})$.

\item There exists a map $\Psi:\mathcal{P}(\mathbf{C}_{1})\times\mathcal{P}(\mathbf{C}_{2})\rightarrow\mathcal{P}(\mathbf{C})$ which
maps pair of non-empty convex subsets
$(C_{1},C_{2})\in\mathcal{P}(\mathbf{C}_{1})\times\mathcal{P}(\mathbf{C}_{2})$
into a nonempty convex subset $C\in\mathcal{P}(\mathbf{C})$, such
that for every $c\in C$, $\phi(c)=(\phi_{1}(c),\phi_{2}(c))\in
C_{1}\times C_{2}$.

\end{enumerate}

\noindent If the map $\Psi$ of the triple compound system
satisfies that for any $c_{1}$ and $c_{2}$,
$\Psi(\{c_{1}\},\{c_{2}\})=\{c\}$ for some $c\in \mathbf{C}$, we
will say that it is a \emph{strictly two-components triple
compound system}.

\end{definition}

\nd With this constructions at hand, let us restrict for the sake
of simplicity to \emph{strictly two-components triple compound
systems} and look for a generalization of entanglement and
separability. It is clear now that the analogues of the maps
$\Lambda$ and $\tau$ are $\Psi$ and $\phi$, respectively. Thus, it
is natural now state

\begin{definition}
Given a \emph{strictly two-components triple compound system}
$(\mathbf{C},\mathbf{C}^{\sharp},u_{\mathbf{C}})$,
$(\mathbf{C}_{1},\mathbf{C}_{1}^{\sharp},u_{\mathbf{C}_{1}})$, and
$(\mathbf{C}_{2},\mathbf{C}^{\sharp}_{2},u_{\mathbf{C}_{2}})$,
with an up-map $\Psi$ and a down-map $\phi$, then

\begin{enumerate}

\item A state $c\in\mathbf{C}$ will be called \emph{non-product state} if
$\Psi\circ\phi(\{c\})\neq\{c\}$. Otherwise, it will be called a
\emph{product state}.

\item For an \emph{invariant convex subset} $C$ one has
$C\in\mathcal{P}(\mathbf{C})$, such that $\Psi\circ\phi(C)=C$

\item If there exist a largest (in the sense of the lattice order) invariant subset, we will denote it by
$\mathcal{S}(\mathbf{C})$.

\item A \emph{strictly two-components triple compound system} for
which there exists $\mathcal{S}(\mathbf{C})$ and  is strictly
included in $\mathbf{C}$, will be said to be an \emph{entanglement
operational model}.

\item In an \emph{entanglement operational model} a state $c$ which satisfies
$c\notin\mathcal{S}(\mathbf{C})$ will be said to be entangled.

\end{enumerate}

\end{definition}

\nd It is clear that using these constructions we can export the
quantum entanglement structure to a much wider class of COM's, and
for that reason, to many new statistical physical systems. It
should be clear also that quantum mechanics is the best example
for entanglement, and that all states in classical mechanics are
separable. Remark that the properties of a \emph{strictly
two-components triple compound systems} will depend, in a strong
sense, on the choice of the functions $\Psi$ and $\phi$. These
should be selected as the canonical ones, i.e., the ones which are
somehow natural for the physics of the problem under study. Notice
that nothing prevents us from make more general choices for
practical purposes. The physical criterium for the construction of
$\psi$ should be that the simple addition of the systems involved
should not generate new correlations. We can also ``postulate" a
generalized separability criterium:

\begin{definition}
\item A state $c\in\mathbf{C}$ in an \emph{entanglement operational model} is said to be separable iff there
exists $C\subseteq\mathcal{S}(\mathbf{C})$ containing $c$ such
that $\Psi\circ\phi(C)=C$.
\end{definition}

\nd These constructions may be useful to develop and search for
generalizations/corrections of quantum mechanics and for the study
of quantum entanglement in theories  more general than quantum
mechanics. Our constructions are a valid  alternative to others
that one can find  in the literature. An interesting open problem
would be that of finding the way in which we can express the
violation of Bell's inequalities using this approach.

\nd In this section we restricted ourselves to \emph{strictly
two-components compound triples}. An important example of a
\emph{two-components compound triple} which is not strict is to
look at a quantal three-components systems out of which we only
consider two subsystems. In that case, to any product state of the
first two subsystems we can add any other state of the third one,
and  the map $\Psi$ will yield a convex subset of more than one
element.

\section{Conclusions}\label{s:Conclusions}

\nd In this work we studied different mathematical structures of
the convex subsets of the quantum set of states. We showed that
these sets are endowed with a canonical lattice structure
 and extended previous results to the infinite dimensional case. This
lattice structure reveals interesting algebraic and geometrical
properties of the quantum set of states.

\nd We showed in Section \ref{s:EntanglementWittness} that the
lattice structure is strongly linked to functionals and entanglement
witness. Thus,  many of previous results might be translated into
 our language. At the end of this Section we also provided a new
(partial) entanglement criteria easily expressible in lattice
theoretical language. We also showed how this algebraic and
geometrical  convex set-viewpoint can be used to reformulate the
Max-Ent principle in a form extensible to any statistical theory,
via the COM approach. In particular, it may be useful to include
fussy measurements (POVM's) into the  Max-Ent formalism.

\nd We also extended a previous abstract separability criterium,
strongly linked to the lattice structure of convex subsets, to the
infinite dimensional case. Furthermore, we showed that this
geometrical setting can be exported to any arbitrary statistical
model via the COM approach, which  may be useful to analyze the
classicality of theories which generalize quantum mechanics, and
also for the study of semiclassical models.

\vskip1truecm

\noindent {\bf Acknowledgements} \noindent This work was partially
supported by the following grants: PIP N$^o$ 6461/05 (CONICET).

\appendix
\section{Basic mathematical concepts used in the
text}\label{s:ApendixA}

\begin{enumerate}

\item A function is surjective (onto) if every possible image is mapped
to by at least one argument. In other words, every element in the
codomain has non-empty preimage. Equivalently, a function is
surjective if its image is equal to its codomain. A surjective
function is a surjection.

\item A linear functional (also called a one-form or covector) is a linear map from a vector space to its field of scalars $K$.
  In general, if $V$ is a vector space over a field $K$, then a linear functional $f$ is a function from $V$ to $K$, which is linear.
 Linear functionals are particularly important in quantum mechanics.  Quantum mechanical systems are represented by Hilbert spaces,
which are anti-isomorphic to their own dual spaces.  A state of  a quantum mechanical system can be identified with a linear functional.

\item  Suppose that $K$ is a field (for example, the real numbers)
and $V$ is a vector space over $K$. If $v_1,\ldots,v_n$ are
vectors and $a_1,\ldots,a_n$ are scalars, then the linear
combination of those vectors with those scalars as coefficients
is, of course, $\sum_{i=1}^n\,a_i\,v_i$.  By restricting the
coefficients used in linear combinations, one can define the
related concepts of affine combination, conical combination, and
convex combination, together with the associated notions of sets
closed under these operations. If $\sum_{i=1}^n \,a_i=1,$ we have an
affine combination, its span being an affine subspace while  the
model space is an hyperplane. If all $a_i \ge 0,$ we have
instead a conical combination, a convex cone and a quadrant,
respectively. Finally, if all $a_i \ge 0$ plus
$\sum_{i=1}^n\,a_i=1,$  we have now a convex combination, a convex set
and a simplex, respectively.

\item By a  $\sigma-$algebra one means  a collection of sets
that satisfy certain properties,  used in the definition of
measures: it is the collection of sets over which a measure is
defined. The concept is important in probability theory, being
there  interpreted as the collection of events which can be
assigned probabilities. Such an algebra, over a set $X$, is a
nonempty collection $S$ of subsets of $X$ (including $X$ itself)
that is closed under complementation and countable unions of its
members. It is an algebra of sets, completed to include countably
infinite operations. The pair $(X, S)$ is also a field of sets,
called a measurable space.

\item A quotient space (also called an identification space) is, intuitively speaking,
the result of identifying certain points of a given space. The
points to be identified are specified by an equivalence relation.
This is commonly done in order to construct new spaces from given
ones. Let $(X, \tau_X)$ be a topological space, and let $R$ be an
equivalence relation on $X$. The quotient space $Y=X/R$ is defined
to be the set of equivalence classes of elements of $X$:
$$Y= \{[x]: x \in X\}=\{\{v\in X: v R x\}: x \in X\},$$ equipped
with the topology where the open sets are defined to be those sets
of equivalence classes whose unions are open sets in $X$.
Equivalently, we can define them to be those sets with an open
pre-image under the quotient map  which sends a point in $X$ to
the equivalence class containing it.

\item Banach spaces are vector spaces $ V$ with a norm $||.|| $  such that every Cauchy sequence
(with respect to the metric $d(x, y) = ||x - y||$ in $V$) has a
limit in $V$ (with respect to the topology induced by that
metric). As for general vector spaces, a Banach space over the
real numbers is called a real Banach space, and a Banach space
over the complex numbers is called a complex Banach space.

\item Algebras: general vector spaces do not possess a multiplication between vectors.
A vector space equipped with an additional bilinear operator
defining the multiplication of two vectors is an algebra over a
field. Many algebras stem from functions on some geometrical
object: since functions with values in a field can be multiplied,
these entities form algebras.

\item In functional analysis, a Banach algebra is an associative algebra $A$
 over the real or complex numbers which at the same time is also a Banach space
    The algebra multiplication and the Banach space norm are required to be related by the following inequality:
    $\forall x, y \in A : \|x \, y\| \ \leq \|x \| \, \| y\|$
(i.e., the norm of the product is less than or equal to the
product of the norms). This ensures that the multiplication
operation is continuous. This property is found in the real and
complex numbers; for instance.

\item A $C^*-$algebra is a Banach algebra with an antiautomorphic involution $*$ which satisfies
$(x^*)^* =   x$ (1); $x^*y^*  =   (yx)^*$ (2); $ x^*+y^* =
(x+y)^*$ (3); and  $(cx)^*  =   c^*\,x^*$ (4), where $c^*$ is the
complex conjugate of $c$, and whose norm satisfies
$||xx^*||=||x||^2$.

\item $C^*$-algebras are an important area of research in
functional analysis. An outstanding example is the complex algebra
 of linear operators on a complex Hilbert space with two additional
properties: \newline \nd it is a topologically closed set in the
norm topology of operators and \newline \nd  is closed under the
operation of taking adjoints of operators. \newline \nd It is
generally believed that these algebras were first considered
primarily for their use in quantum mechanics to model algebras of
physical observables, beginning with Werner Heisenberg's matrix
mechanics and  developed further by Pascual Jordan circa 1933.
Afterwards, John von Neumann  established a general framework for
them which culminated in papers on rings of operators,  considered
as a special class of C*-algebras  known as von Neumann algebras.

\item It is now generally accepted that the description of quantum mechanics in which all
self-adjoint operators represent observables is untenable. For
this reason, observables are identified to elements of an abstract
C*-algebra $A$ (that is one without a distinguished representation
as an algebra of operators) and states are positive linear
functionals on $A$. However, by using the GNS construction, we can
recover Hilbert spaces which realize $A$ as a subalgebra of
operators. Geometrically, a pure state on a C*-algebra $A$ is a
state which is an extreme point of the set of all states on $A$.
By properties of the GNS construction these states correspond to
irreducible representations of $A$. The states of the C*-algebra
of compact operators $K(\mathcal{H})$ correspond exactly to the
density operators and therefore the pure states of
$K(\mathcal{H})$  are exactly the pure states in the sense of
quantum mechanics. The C*-algebraic formulation can be seen to
include both classical and quantum systems. When the system is
classical, the algebra of observables become an abelian
C*-algebra. In that case the states become probability measures.

\item In functional analysis,  given a C*-algebra $A$, the Gelfand-Naimark-Segal
(GNS) construction establishes a correspondence between cyclic
*-representations of $A$ and certain linear functionals on $A$
(called states). The correspondence is shown by an explicit
construction of the *-representation from the state.

\item A *-representation of a C*-algebra $A$ on a Hilbert space
$\mathcal{H}$ is a mapping $\pi$ from $A$ into the algebra of
bounded operators on $\mathcal{H}$.

\item  Point-wise convergence is one of various senses in which
a sequence of functions can converge to a particular function.
Suppose $\{ f_n \}$ is a sequence of functions sharing the same
domain and codomain. The sequence $\{ f_n \}$ converges pointwise
to $f$, often written as $\lim_{n \rightarrow \infty} f_n=f$ point
wise iff for every $x$ in the domain one has $\lim_{n \rightarrow
\infty} f_n(x)=f(x)$.

\item Every subset $Q$ of a vector space is contained within a
smallest convex set (called the convex hull of $Q$), namely the
intersection of all convex sets containing $Q$, \item A set with a
binary relation $R$ on its elements that is reflexive (for all $a$
in the set, $aRa$), antisymmetric (if $aRb$ and $bRa$, then $a =
b$) and transitive (if $aRb$ and $bRc$, then $aRc$) is described
as a partially ordered set or poset, \item  Let $X$ be a space.
Its dual space $X^*$ consists of all linear functions from $X$
into the base field $K$ which are continuous with respect to the
prevailing topology. \item The {\it weak} topology on $X$  is the
coarsest topology (the topology with the fewest open sets) such
that each element of $X^*$ is a continuous function. \item The
predual of a space $D$ is a space $D'$ whose dual space is $D$.
For example, the predual of the space of bounded operators
$\mathcal{B}(\mathcal{H})$ is the space of trace class operators,
 \item The ultraweak topology, also called the weak-* topology,
on the set $\mathcal{B}(\mathcal{H})$ is the weak-topology
obtained from the trace class operators on $\mathcal{H}$. In other
words it is the weakest topology such that all elements of the
predual are continuous (when considered as functions on
$\mathcal{H}$),
\item A partially-ordered group is a group $(G,+)$ equipped with a partial order
``$\vdash$" that is translation-invariant. That is, ``$\vdash$"
has the property that, for all $a$, $b$, and $g$ in $G$, if $a
\vdash b$ then $a+g \vdash b+g$ and $g+a \vdash g+b$, \item  An
element $x$ of $G$ is called positive element if $0 \vdash x$. The
set of elements $0 \vdash x$ is often denoted with $G+$, and it is
called the {\it positive cone} of $G$. So we have $a \vdash b$ if
and only if $-a+b \in G+$. \item  For the general group $G$, the
existence of a positive cone specifies an order on $G$. A group
$G$ is a partially-ordered group if and only if there exists a
subset $J$ (which is $G+$) of $G$ such that: $0 \in J$;  if $a \in
J$ and $b \in J$ then $a+b \in J$; if $a \in J$ then $-x+a+x \in
J$ for each $x$ of $G$; if $a \in J$ and $-a \in J$ then $a \vdash
0$.

\item In linear algebra, a matrix decomposition is a factorization of a matrix into some canonical form.
There are many different matrix decompositions; each finds use
among a particular class of problems.
  {\it The Cholesky decomposition} is applicable to any square, symmetric, positive definite matrix
  $A$ in the form $A = U^T\,U$, where $U$ is upper triangular with positive
diagonal entries. The Cholesky decomposition is a special case of
the symmetric LU decomposition, with $L = U^T$. The Cholesky
decomposition is unique and also  applicable for complex hermitian
positive definite matrices. {\it The singular value decomposition}
 is applicable to $m$ times $n$ matrix $A$ in the fashion $A = UDV^{\dagger}$, where
 $D$ is a nonnegative diagonal matrix while $U$, $V$ are unitary
matrices, and $V^{\dagger}$ denotes the conjugate transpose of $V$
(or simply the transpose, if $V$ contains real numbers only). The
diagonal elements of $D$ are called the singular values of $A$.

\item The orthogonal complement $W^{\bot}$ of a subspace $W$ of an inner product space
$V$ is the set of all vectors in $V$ that are orthogonal to every vector in $W$, i.e.,

   $$ W^\bot=\left\{x\in V : \langle x, y \rangle = 0 \mbox{ for all } y\in W \right\}.$$

\item A topological space is called {\it separable} if it contains a countable dense subset. In other words,
 there exists a sequence $\{ x_n \}_{n=1}^{\infty}$  of elements of the space such that
every nonempty open subset of the space contains at least one element of the sequence.

\item A cover of a set $X$ is a collection of sets whose union contains $X$ as a subset.

\item A topological space $X$ is called compact if each of its open covers has a finite subcover. Otherwise it is called non-compact.

\item A relatively compact subspace (or relatively compact subset) $Y$ of a topological space $X$ is a subset whose closure is compact.

\item   $T$  is a compact operator on Hilbert's space if the image of each bounded set under $T$ is relatively compact. \newline \nd
Compact operators on Hilbert spaces are a direct extensions of matrices.
In such spaces  they are the closure of finite-rank operators. As such, results from matrix theory can sometimes be extended
to compact operators using similar arguments.
In contrast, the study of general operators on infinite dimensional spaces often requires a genuinely different approach.
For example, the spectral theory of compact operators on Banach spaces takes a form that is very similar to the Jordan canonical form of matrices. In the context of Hilbert spaces, a square matrix is unitarily diagonalizable if and only if it is normal. A corresponding result holds for normal compact operators on Hilbert spaces.

\item A simplex is a generalization of the notion of a triangle or tetrahedron to arbitrary dimension.
Specifically, an $n-$simplex is an $n-$dimensional polytope which is the convex hull of its $n + 1$ vertices.
A 2-simplex is a triangle, a 3-simplex is a tetrahedron, and a 4-simplex is a pentachoron. A single point may be considered a 0-simplex, and a line segment may be considered a 1-simplex. {\it A simplex may be defined as the smallest convex set containing the given vertices}.
\end{enumerate}


\section{Lattices}\label{s:ApendixB}

\noindent A {\sf lattice} $\mathcal{L}$    (also called a poset)
is a partially ordered set (also called a poset) in which any two
elements $a$ and $b$ have a unique supremum (the elements' least
upper bound ``$a\vee b$"; called their join) and an infimum
(greatest lower bound ``$a\wedge b$"; called their meet). Lattices
can also be characterized as algebraic structures satisfying
certain axiomatic identities. Since the two definitions are
equivalent, lattice theory draws on both order ($>$, $<$) theory
and universal algebra. Semilattices include lattices, which in
turn include Heyting and Boolean algebras. These ``lattice-like"
structures all admit order-theoretic as well as algebraic
descriptions.

\vskip 3mm \nd A {\it bounded} lattice has a greatest (or maximum)
and least (or minimum) element, denoted $1$ and $0$ by convention
(also called top and bottom, respectively). Any lattice can be
converted into a bounded lattice by adding a greatest and least
element, and every non-empty finite lattice is bounded.     For
any set $A$, the collection of all subsets of $A$ (called the
power set of $A$) can be ordered via subset inclusion to obtain a
lattice bounded by $A$ itself and the null set. Set intersection
and union represent the operations meet and join, respectively.

\vskip 3mm \nd A poset is called a complete lattice if all its
subsets have both a join and a meet. In particular, every complete
lattice is a bounded lattice. While bounded lattice homomorphisms
in general preserve only finite joins and meets, complete lattice
homomorphisms are required to preserve arbitrary joins and meets.

\vskip 3mm \nd  Any quantum system represented by an
$N-$dimensional Hilbert space $\mathcal{H}$ has associated a
lattice formed by all its convex subspaces
${\mathcal{L}}_{v\mathcal{N}}({\mathcal{H}})=
<{\mathcal{P}}({\mathcal{H}}),\ \cap,\ \oplus,\ \neg,\ 0,\ 1>$,
where $0$ is the empty set $\emptyset$, $1$ is the total space
$\mathcal{H}$,  $\oplus$ the closure of the sum, and $\neg(S)$ is
the orthogonal complement of a subspace $S$
\cite{mikloredeilibro}. This lattice was called ``Quantum Logic"
by Birkhoff and von Neumann. One  refers to this lattice as the
von Neumann-lattice $\mathcal{L}_{v\mathcal{N}}(\mathcal{H})$)
\cite{mikloredeilibro}.

\vskip 3mm

 \nd  Let $\mathcal{L}$ be a bounded lattice with greatest element 1 and least element 0.
 Two elements $x$ and $y$ of the lattice are complements of each other if and only
 if:  $x\bigvee y=1$
 and $x\bigwedge  y=0$.
In the case the complement is unique, we write $\neg x = y$ and
equivalently, $\neg y = x$. A bounded lattice for which every
element has a complement is called a {\it complemented} lattice.
The corresponding unitary operation over the lattice, called
complementation, introduces an analogue of logical negation into
lattice theory. The complement is not necessarily unique, nor does
it have a special status among all possible unitary operations over
 $\mathcal{L}$.

\vskip 3mm \nd Distributive lattices are lattices for which the operations of join and meet distribute over each other.
The prototypical examples of such structures are collections of sets for which the lattice operations can be given by set
union and intersection.
Indeed, these lattices of sets describe the scenerio completely.
A complemented lattice that is also distributive is a Boolean
algebra. For a distributive lattice, the complement of $x$, when
it exists, is unique.

\vskip 3mm

 \nd  The concept of  lattice's atom is of great physical importance.  If
$\mathcal{L}$ has a null element $ 0$, then an element $x$ of
$\mathcal{L}$ is an {\it atom}  if $0 < x$ and there exists no
element $y$ of $\mathcal{L}$ such that $0 < y < x$. One  says that
$\mathcal{L}$ is: \newline i) {\it Atomic}, if for every nonzero
element $x$ of $\mathcal{L}$, there exists an atom $a$ of
$\mathcal{L}$ such that $ a = x$
\newline ii) Atomistic, if every element of $\mathcal{L}$ is a
supremum of atoms. \vskip 3mm

 \nd A  {\it modular} lattice is one that
satisfies the following self-dual condition (modular law) $ x \leq
b$ implies $x \vee (a \wedge b) = (x \vee a) \wedge b$, where
$\le$ is the partial order, and $\vee$ and $\wedge$ (join and
meet, respectively) are the operations of the lattice. \vskip 3mm

\nd Modular lattices arise naturally in algebra and in many other
areas of mathematics. For example, the subspaces of a vector space
(and more generally the submodules of a module over a ring) form a
modular lattice. Every distributive lattice is modular. In a not
necessarily modular lattice, there may still be elements $b$ for
which the modular law holds in connection with arbitrary elements
$a$ and $x$ ($\le  b$). Such an element is called a modular
element. Even more generally, the modular law may hold for a fixed
pair $(a, b)$. Such a pair is called a modular pair, and there are
various generalizations of modularity related to this notion and
to semi-modularity. \vskip 3mm

 \nd For $a,\,b \in \mathcal{L}$,
to assert that $a$ is orthogonal to $b$ ($ a \bot b$)  implies $a
\wedge b=0$. Equivalently, in ``order" terms, one says that $a\le
b^{\bot}$. Now, $\mathcal{L}$ is an orthocomplemented lattice if
whenever $a \bot b$ then $b\le a^{\bot}$. $\bot$ is a symmetric
relation.

\vskip 3mm

\nd For any $a \in \mathcal{L}$, define $M(a):
 =\{c\in \mathcal{L} \vert c\bot a, \, {\rm and}\, 1=c\vee a\}$. An element of $M(a)$ is
called an orthogonal complement of $a$. We have $a^{\bot} \in
M(a)$, and any orthogonal complement of $a$ is a complement of
$a$. If we replace the unity in $M(a)$ by an arbitrary element $b
\ge  a$, then we have the set $M(a,b):=\{c\in \mathcal{L}\vert c
\vee a \,{\rm and}\, b=c\vee a\}$. An element of $M(a,b)$ is
called an orthogonal complement of $a$ relative to $b$ . Clearly,
$M(a)=M(a,1)$. Also, for $a\le cb$ , $c\in M(a,b)$, iff $a\in
M(c,b)$. As a result, we can define still another symmetric binary
operator $\oplus$ on $[0,b]$, given by $b=a \oplus c$ iff $c\in
M(a,b)$. Note that $b=b \oplus 0$.  A final operation is the
``difference" $b-a=b\wedge a$. Some properties: (1) $a-a=0$,
$a-0=a$, $0-a=0$, $a-1=0$ , and $1-a=a^{\bot}$; (2) $b-a=a-b$; (3)
if $a\le b$, then $a\wedge (b-a)$ and $a \oplus (b-a)  \le b$.

\vskip 3mm

\nd {\sf Definition:} A lattice $\mathcal{L}$ is called an
orthomodular lattice if i) $\mathcal{L}$ is orthocomplemented, and
(orthomodular law) ii)  if  $x\le y$, then $y=x\oplus (y-x)$. The
orthomodular law can be recasted as follows: if $x\le y$ , then
$y=x\vee (y \wedge x^{\bot})$. Equivalently, $x\le y$ implies
$y=(y\wedge x)\vee (y\wedge x^{\bot})$. Such relation is
automatically true in an arbitrary distributive lattice, even
without the assumption that $x\le y$. For example, the lattice
$\mathcal{L}(H)$  of closed subspaces of a Hilbert space $H$ is
orthomodular. $\mathcal{L}(H)$ is modular iff $H$ is finite
dimensional. In addition, if we give the set $\mathcal{P}_p(H)$ of
(bounded) projection operators on $H$ an ordering structure by
defining $P\le Q$ iff $\mathcal{P}(H) \le \mathcal{Q}(H)$, then
 $\mathcal{P}_p(H)$is lattice isomorphic to $\mathcal{L}(H)$, and
hence orthomodular \cite{BvN}.

\section{Faces of a convex set}\label{s:ApendixC}

\noindent We define here a convex set's {\it face} in a real vector
space of finite dimension. Let $\mathcal{C}$ be a convex subset of
$\mathbb{R}^n$ and  let us introduce the auxiliary notions of
oriented hyperplanes and supporting hyperplanes. Given ${\bf
n,p}\in\mathbb{R}^n$ let us define the hyperplane $H( {\bf n},  {\bf
p})$ via
$$H( {\bf n},  {\bf p}) = \{{\bf x} \in \mathbb{R}^n: {\bf n} \cdot ({\bf x}  -    {\bf p})=0\}.$$
If ${\bf n}=0$ it is equal to $\mathbb{R}^n$ and we call it
degenerate. As long as $H( {\bf n},  {\bf p})$ is nondegenerate, its
removal disconnects $\mathbb{R}^n$. The upper halfspace of
$\mathbb{R}^n$ determined by $H( {\bf n},  {\bf p})$ is $H( {\bf n},
{\bf p})^+=\{{\bf x} \in \mathbb{R}^n:  {\bf n} \cdot ({\bf x}  -
{\bf p})   \ge 0\}.$ A hyperplane $H( {\bf n},  {\bf p})$
\underline{is a supporting hyperplane} for $\mathcal{C}$ if its
upper halfspace contains $\mathcal{C}$, that is, if $\mathcal{C}
\subset  H( {\bf n}, {\bf p})^+$.

\vskip 3mm \noindent Using this terminology, we can define a {\bf
face} of a convex set $\mathcal{C}$ to be the intersection of
$\mathcal{C}$ with a supporting hyperplane of $\mathcal{C}$. Notice
that we still get both the empty set and $\mathcal{C}$ itself  as
improper faces of $\mathcal{C}$. For the definition of a face in the
infinite dimensional case we extend the definition of a supporting
hyperplane to a real Hilbert space $\mathcal{H}$. Given ${\bf
n,p}\in\mathcal{H}$, we say that $H( {\bf n},{\bf p})$,
$$H( {\bf n},{\bf p})=\{{\bf x} \in \mathcal{H}: \langle{\bf n} ,{\bf x}-{\bf p}\rangle=0\},$$
is a supporting hyperplane if $\mathcal{C} \subset  H( {\bf n}, {\bf
p})^+$. Note that $H( {\bf n},{\bf p})$ is closed and using Riesz
representation theorem, for every closed hyperplane $H$
there exists ${\bf n,p}\in\mathcal{H}$ such that $H=H( {\bf n},{\bf p})$.\\

\noindent In the general case (in a Banach space) we say that
\underline{$F$ is a face of $\mathcal{C}$} if there exist a
\underline{closed} hyperplane $H$ such that $F=\mathcal{C}\cap H$.
A closed hyperplane is given by a continuos lineal functional.\\

\noindent {\bf Remarks:} Let $\mathcal{C}$ be a convex set. Then:

\begin{itemize}
\item  If $F_1=\mathcal{C}\cap H( {\bf n_1},  {\bf p_1})$ and $F_2=\mathcal{C}\cap H( {\bf n_2},  {\bf p_2})$ are faces of
$\mathcal{C}$ intersecting at a point $p$ then $H({\bf n}_1+ {\bf
n}_2, {\bf p})$ is a supporting hyperplane of $\mathcal{C}$ and
$F1\cap F2=C\cap H({\bf n}_1+ {\bf n}_2, {\bf p})$. This shows that
the faces of $\mathcal{C}$ form a meet-semilattice.

\item  Since each proper face lies on the base of the upper halfspace of some supporting hyperplane,
each such face must lie on the relative boundary of $\mathcal{C}$.
\end{itemize}

\noindent An extreme point of a convex set $\mathcal{C}$ in a real
vector space is a point in $\mathcal{C}$ which does not lie in any
open line segment joining two points of $\mathcal{C}$. Intuitively,
an extreme point is a "corner" of $\mathcal{C}$. The Krein-Milman
theorem states that if $\mathcal{C}$ is convex and compact in a
locally convex space, then $\mathcal{C}$ is the closed convex hull
of its extreme points.  In particular, such a set has extreme
points.\\

\end{document}